\documentclass[a4paper]{article}
\usepackage[utf8]{inputenc}
\usepackage{amsmath,amssymb,amsthm}
\usepackage{graphicx,tikz}
\usepackage{placeins}
\usepackage{subcaption} 
\usepackage{xpatch}
\usepackage{apptools}
\usepackage{todonotes}
\usepackage{comment}
\usepackage{thmtools,thm-restate}
\usepackage{algpseudocode}
\usepackage{soul}

\usepackage[linesnumbered,ruled]{algorithm2e}
\usepackage{float}
\usetikzlibrary{chains,fit,shapes}
\newtheorem{theorem}{Theorem}
\newtheorem{lemma}{Lemma}
\newtheorem{definition}{Definition}
\newtheorem{corollary}{Corollary}

\newcounter{claimctr}[lemma]
\newcounter{remctr}

\newenvironment{remark}{\medskip\par\noindent\refstepcounter{remctr}\textbf{Remark \theremctr.} }{\smallskip\par}

\newtheorem{observation}{Observation}
\newtheorem{question}{Question}
\newenvironment{myclaim}{\refstepcounter{claimctr}\medskip\par\noindent\textit{Claim \theclaimctr.}}{\par}
\newenvironment{myenumerate}
{\vspace{-0.05in}
	\begin{enumerate}
		\itemsep -0.04in

	}
	{\end{enumerate}}
\makeatletter
\xpatchcmd{\thmt@restatable}
{\csname #2\@xa\endcsname\ifx\@nx#1\@nx\else[{#1}]\fi}
{\IfAppendix{\csname #2\@xa\endcsname}{\csname #2\@xa\endcsname\ifx\@nx#1\@nx\else[{#1}]\fi}}
{}{} 
\makeatother

\newcommand{\com}[1]{}

\title{$\chi$-binding functions for squares of bipartite graphs and its subclasses}

\author{Dibyayan Chakraborty\footnote{School of Computing, University of Leeds, United Kingdom} \and L. Sunil Chandran\footnote{Indian Institute of Science, Bengaluru, India} \and Dalu Jacob\footnotemark[2] \and Raji R. Pillai\footnotemark[2]}
\date{}
\newcommand{\col}[1]{\chi\left( #1\right)}

\newcommand{\clique}[1]{\omega\left( #1\right)}

\newcommand{\squares}[1]{\left(#1\right)^2}

\newcommand{\lowerbound}[1]{\Omega\left( #1\right)}

\newcommand{\leftof}{\textit{left}}
\newcommand{\rightof}{\textit{right}}
\usepackage{geometry}
\begin{document}

	\maketitle
	
	\begin{abstract}
		A class of graphs $\mathcal{G}$ is $\chi$-bounded if there exists a function $f$ such that $\chi(G) \leq f(\omega(G))$ for each graph $G \in \mathcal{G}$, where $\chi(G)$ and $\omega(G)$ are the chromatic and clique number of $G$, respectively. 
		The \emph{square} of a graph $G$, denoted as $G^2$, is the graph with the same vertex set as $G$ in which two vertices are adjacent when they are at a distance at most two in $G$. In this paper, we study the $\chi$-boundedness of \emph{squares} of bipartite graphs and its subclasses. Note that the class of squares of graphs, in general, admit a quadratic $\chi$-binding function. 
        Moreover there exist bipartite graphs $B$ for which  $\col{B^2}$ is $\lowerbound{\frac{\squares{\clique{B^2}} }{\log \clique{B^2}}}$. We first ask the following question: ``What sub-classes of bipartite graphs have a linear $\chi$-binding function?" We focus on the class of \emph{convex bipartite} graphs and prove the following result: for any convex bipartite graph $G$, $\col{G^2} \leq \frac{3 \clique{G^2}}{2}$. Our proof also yields a polynomial-time $3/2$-approximation algorithm for coloring squares of convex bipartite graphs. We then introduce a notion called ``partite testable properties" for the squares of bipartite graphs. We say that a \emph{graph property $P$ is partite testable for the squares of bipartite graphs if for a bipartite graph $G=(A,B,E)$, whenever the induced subgraphs $G^2[A]$ and $G^2[B]$ satisfies the property $P$ then $G^2$ also satisfies the property $P$}. Here, we discuss whether some of the well-known graph properties like \emph{perfectness}, \emph{chordality}, \emph{(anti-hole)-freeness}, etc. are partite testable or not. As a consequence, we prove that the squares of biconvex bipartite graphs are perfect. 

         \end{abstract}

	\section{Introduction}
	All the graphs we consider in the paper are finite, undirected, and simple. For a graph $G$, we denote the vertex set as $V(G)$ and the edge set as $E(G)$. A $k$-coloring of a graph $G$ is a mapping $c\colon V(G) \rightarrow \{1,2,\ldots,k\}$ such that for any edge $uv\in E(G)$, $c(u) \neq c(v)$. The \emph{chromatic number} of $G$, denoted by $\col{G}$, is the minimum integer $k$ such that $G$ admits a $k$-coloring. The \emph{clique number} of a graph $G$, denoted by $\clique{G}$, is the largest integer $k$ such that $G$ contains a complete subgraph on $k$ vertices.  A class of graphs $\mathcal{G}$ is $\chi$-bounded if there exists a function $f$ such that $\chi(G) \leq f(\omega(G))$ for each graph $G \in \mathcal{G}$, where $\chi(G)$ and $\omega(G)$ are the chromatic and clique number of $G$, respectively. The study on $\chi$-boundedness of graph classes has recently gained popularity among researchers in the field. See the survey by Scott and Seymour~\cite{scott2020survey}.
	
	\medskip
	
	In this paper, we study the $\chi$-boundedness of \emph{squares} of bipartite graphs. The \emph{square} of a graph $G$, denoted as $G^2$, is the graph with the same vertex set as $G$ in which two vertices are adjacent when they are at a distance at most two in $G$.  A graph $H$ is a \emph{square} graph if there exists a graph $G$ such that $G^2\cong H$. Note that coloring square graphs is also known as \emph{distance $2$-coloring} of graphs~\cite{KRAMER2008422} and has applications in broadcast scheduling in multi-hop radio networks, distributed computing, etc.~\cite{fraigniaud2020,congestmodel}.
	
	\medskip
	Bounding the chromatic number of graph squares has been receiving attention over the past 25 years.  See a recent survey by Cranston~\cite{cranston2022}. 
 
 Coloring squares of graphs is interesting, even for special classes of graphs. Wegner’s conjecture on the chromatic number of squares of planar graphs is a perfect example of this. In 1977, Wegner conjectured the following: Let $G$ be a planar graph with maximum degree $\Delta$. If $\Delta = 3$, then $\chi(G^2)\leq 7$, if $4 \leq \Delta \leq 7$, then $\chi(G^2)\leq \Delta+5$, and  if $\Delta \geq 8$, then $\chi(G^2)\leq \big\lfloor \frac{3\Delta}{2}\big\rfloor+1$.  
	Wegner's conjecture for $\Delta = 3$ has been confirmed by Thomassen~\cite{thomassen2018}. But the conjecture is still open for $\Delta \geq 4$. Due to the popularity of Wegner's conjecture, several upper bounds in terms of maximum degree are known for the chromatic number of squares of planar graphs~\cite{agnarsson2003coloring,molloy2005bound,van2003coloring,zhu2022}: in particular for planar graphs with high girth~\cite{choi2020,dong2019,dvovrak2008}. Apart from planar graphs, there are several other special classes of graphs, including cocomparability graphs, circular-arc graphs, etc., for which the bounds on the chromatic number of their squares are investigated~\cite{calamoneri2009}.  

	\medskip

	\medskip
	Observe that if the maximum degree of a graph $G$ is $\Delta$, then the maximum degree of $G^2$ is at most $\Delta^2$, and therefore
 we have $\col{G^2}\leq \Delta^2+1$. On the other hand, we can observe that for any vertex $v\in V(G)$, the vertex $v$, together with its neighborhood in $G$, form a clique in $G^2$. This implies that $\omega(G^2)\geq \Delta+1$. Together with the previous observation, we then have $\chi(G^2)\leq (\omega(G^2)-1)^2+1$. Hence, the class of square graphs is $\chi$-bounded with a quadratic $\chi$-binding function. Alon and Mohar~\cite{Alon2002} have studied the upper and lower bounds for the chromatic number of squares of graphs in terms of maximum degree. Since bipartite graphs are the class of graphs having the least chromatic number, it would be interesting to look at the chromatic number of the squares of bipartite graphs. In the following section, we summarize our results and their significance.

\subsection{Our results}
	
	 First, we observe some implications of the result of Alon and Mohar~\cite{Alon2002}. In particular, we show the existence of bipartite graphs $B$ such that the chromatic number of $B^2$ is $\Omega(\frac{\omega^2}{\log\omega})$, where $\omega$ is the clique number of $B^2$ (Observation~\ref{obs:bipartite}). 
 Furthermore, we also observe that if the girth of a graph $G$ is greater than six, then the chromatic number of $G^2$ is  $O(\frac{\omega^2}{\log\omega})$, where $\omega$ is the clique number of $G^2$ (Observation~\ref{obs:general}). 
	
	\smallskip
	The above observation implies that the squares of bipartite graphs with girth greater than six have a sub-quadratic $\chi$-binding function. But the question of \emph{whether the squares of (bipartite) graphs, in general, have a sub-quadratic $\chi$-binding function or not}, remains open. In this paper, we are interested in finding sub-classes of bipartite graphs, say $\mathcal{C}$ for which $\mathcal{C}^2$ (i.e. the class of graphs obtained by taking squares of graphs in $\mathcal{C}$) has a linear $\chi$-binding function. Specifically, we investigate the following question.

	\begin{question}
		Identify some popular sub-classes of bipartite graphs whose squares admit a linear $\chi$-binding function. 
	\end{question}

	Let $G=(A,B,E)$ be a bipartite graph. Observe that while taking the square of $G$, the set of edges whose endpoints are in different partite sets remain unchanged in $G^2$, and the additional edges in $G^2$ appear only between the vertices belonging to the same partite set. Also, note that $\chi(G^2)\leq \chi(G^2[A])+\chi(G^2[B])$. Therefore, the structure of the graphs induced by the partite sets, $G^2[A]$ and $G^2[B]$, can have an influence on the chromatic number of $G^2$. A graph $H$ is \emph{perfect} if for every induced subgraph $H'$ of $H$ we have $\col{H'} = \clique{H'}$. Since we are particularly interested in finding the sub-classes of bipartite graphs that admit a linear $\chi$-binding function, the bipartite graphs $G$ for which $G^2[A]$ and $G^2[B]$ are \emph{perfect} is a natural candidate. Let $\mathcal{C}$ be the class of bipartite graphs $G=(A,B,E)$ such that \emph{both $G^2[A]$ and $G^2[B]$ are perfect}. Then observe that for every graph $H$ in $\mathcal{C}^2$, we have $\col{H}\leq 2\clique{H}$. Interestingly, the class of \emph{convex bipartite} graphs is one such class of bipartite graphs. 

 
 
	
	\begin{definition}[Convex bipartite graph]\label{def:convex}
		A bipartite graph $G=(A,B,E)$ is said to be convex bipartite if the vertices in $B$ have an ordering such that, for each vertex $a\in A$, the vertices in the neighborhood of $a$ in $G$ appear consecutively with respect to the ordering.
	\end{definition}
 
Let $\mathcal{C}$ denote the class of convex bipartite graphs and $G=(A,B,E)$ be a graph in $\mathcal{C}$. Note that both $G^2[A]$ and $G^2[B]$ are \emph{interval graphs} (a subclass of perfect graphs) in $G^2$~\cite{Le2019hardness}. But, the graphs in $\mathcal{C}^2$ need not be perfect. For example, the graph $G$ given in Figure~\ref{fig:not_perfect} is a convex bipartite graph for which $G^2$ contains  a $C_5$, namely, $(v_1,v_2,v_3,v_4,v_5)$ as an induced subgraph (edges of $C_5$ are shown in red), and therefore, $G^2$ is not perfect. Therefore, it would be interesting to see whether the trivial upper bound of $2\omega(G^2)$ for $\chi(G^2)$ can be improved for $G\in \mathcal{C}$. 
Note that in the literature, we can see similar instances of refining the trivial $\chi$-binding function $2\omega$ to $3\omega/2$ for other special graph classes as well. The \emph{class of circular-arc graphs} (intersection graphs of a set of arcs on a circle) is such an example~\cite{tucker1975coloring,karapetian1980coloring,valencia2003revisiting}.

 
 \medskip
 
 The following is one of the main results of this paper. 
	\begin{theorem} \label{thm:theorem-convex-bipartite}
		Let $G$ be a convex bipartite graph, and let $\omega(G^2)$ denote the size of a maximum clique in $G^2$. Then $\chi(G^2)\leq \big\lfloor \frac{3\omega(G^2)}{2}\big\rfloor$. Moreover, there exists a convex bipartite graph $H$, such that $\chi(H^2)\geq \frac{5\omega(H^2)}{4} - 2$. 
	\end{theorem}

 Note that the squares of convex bipartite graphs may contain $K_{1,t}$ (for arbitrarily large $t$), or arbitrarily large induced \emph{cliques, paths, cycles or wheels} with odd or even parities, structures like \emph{pyramid, diamond, gem, $2K_2$, paraglider, chair} etc. as induced subgraphs. Therefore, existing results (e.g., ~\cite{karthick2018coloring,karthick2019,kloks2009even,valerio,scott2020survey}) that guarantee a $3\omega/2$ $\chi$-binding function for graphs that do not contain some of the above structures as induced subgraphs, cannot be applied directly to prove Theorem~\ref{thm:theorem-convex-bipartite}. In 1879, while attempting to
	prove the Four Colour Theorem, Kempe~\cite{kempe1879} introduced an elementary operation on 
	graph coloring that later became known as a \emph{Kempe change}.  Let $c$ be a $k$-coloring of a graph $G$. Then, a \emph{Kempe chain} is a maximal bichromatic component. A \emph{Kempe change} in $c$ is equivalent to swapping the two colors
	in a Kempe chain to produce a different $k$-coloring of $G$. 
	
	\smallskip
	We use this well-known Kempe change technique as a tool for proving Theorem~\ref{thm:theorem-convex-bipartite}.  Further, for a graph class, $\mathcal{C}$, as noted by Scott and Seymour~\cite{scott2020survey}, the following question is also interesting. 
 \begin{question}
     For a graph class $\mathcal{C}$, is there a polynomial-time algorithm for coloring graphs $G\in \mathcal{C}$ that uses only $f(\omega(G))$ colors, where $f$ is the $\chi$-binding function for $\mathcal{C}$?
 \end{question}
 Recall that we prove in Theorem~\ref{thm:theorem-convex-bipartite} that the class of squares of convex bipartite graphs is $3\omega/2$-colorable. In fact, our proof also yields a $3/2$-approximation algorithm for coloring the squares of convex bipartite graphs and thereby answers the above question for the same class. Further, we note that for a convex bipartite graph $G$, the function relating $\chi(G^2)$ and $\omega(G^2)$ is not the same as the function relating $\chi(G^2)$ and $\Delta(G)$. In particular, if $G$ is a convex bipartite graph with maximum degree $\Delta$, we prove that $\chi(G^2)\leq 2\Delta$. Moreover, there exist convex bipartite graphs with $\chi(G^2)=2\Delta$ (Theorem~\ref{thm:degree}).
	
	\medskip
The study on the structure of squares of convex bipartite graphs motivated us to introduce the notion of \emph{partite testable properties} for the squares of bipartite graphs, which is defined as follows.
\begin{definition}[Partite testable property] \label{def:partitetest}
    Let $\mathcal{C}$ be a class of bipartite graphs and let $P$ be a graph property. We say that $P$ is a partite testable property for $\mathcal{C}^2$, if for any bipartite graph $G=(A,B,E)$ in $\mathcal{C}$ whenever the induced subgraphs $G^2[A]$ and $G^2[B]$ satisfies the property $P$ then $G^2$ also satisfies the property $P$.
\end{definition}
For instance, in Theorem~\ref{thm:oddantihole}, we prove that \emph{the property of not containing odd anti-holes of length greater than five is a partite testable property for the squares of bipartite graphs in general}. i.e. for a bipartite graph $G=(A,B,E)$ if $G^2[A]$ and $G^2[B]$ do not contain odd anti-holes of length greater than five, then $G^2$ does not contain odd anti-holes of length greater than five.

\smallskip

 Now, a natural question is whether \emph{perfectness} is a partite testable property (in general) or not. 
 Note that a graph $G$ is perfect if and only if it does not contain odd holes or odd anti-holes as induced subgraphs. As mentioned above, the property of \emph{not containing odd anti-holes of length greater than five} is a partite testable property for the squares of bipartite graphs. However, the property of being \emph{perfect} fails to be a partite testable property even for the squares of convex bipartite graphs, because \emph{(odd-hole)-freeness} is not a partite testable property 
 for the class of squares of convex bipartite graphs. To see this, recall that for a convex bipartite graph $G=(A,B,E)$, the subgraphs $G^2[A]$ and $G^2[B]$ are both \emph{interval graphs} (a subclass of perfect graphs, and hence, (odd hole)-free)~\cite{Le2019hardness}, whereas the graph $G^2$ is \emph{not necessarily (odd hole)-free} and hence, not always \emph{perfect} (see Figure~\ref{fig:not_perfect}). 

 \medskip

On the other hand, we find some interesting subclasses of squares of convex bipartite graphs for which \emph{perfectness is a partite testable property} (see Theorem~\ref{thm:biconvex} and Theorem~\ref{thm:c4free}).  Such a graph class includes \emph{$C_5$-free squares of convex bipartite graphs} (i.e. graphs $G^2$ such that $G^2$ is $C_5$-free and $G$ is a convex bipartite graph). 
Further, we observe that for the class of \emph{$C_4$-free squares of convex bipartite graphs}, even \emph{chordality is a partite testable property}. The above results are also interesting due to the fact that the above subclasses are not hereditary. For example, an induced subgraph of the square of a convex bipartite graph need not be the square of a convex bipartite graph. Even though the notion of perfectness is not just limited to hereditary graph classes, most of the well-known classes of perfect graphs are hereditary. Theorem~\ref{thm:c4free} contributes the class of graphs, namely, \emph{$C_5$-free squares of convex bipartite graphs} 
to the rare collection of non-hereditary perfect graphs.
	
	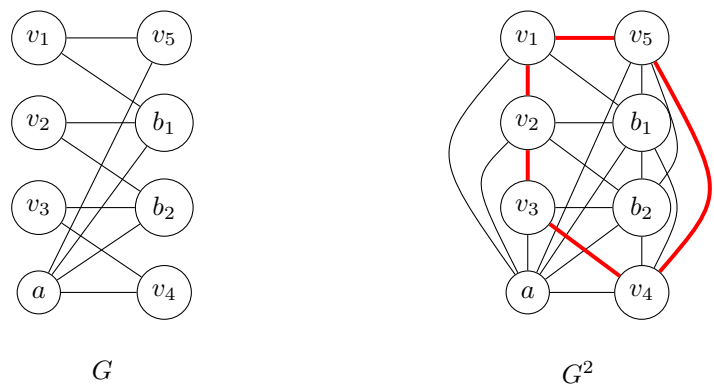
\begin{figure}[h]
		\begin{tabular}{p{.4\textwidth}p{.4\textwidth}}
			\parbox{.4\textwidth}{\centering
				\begin{tikzpicture}[scale=.75]
					\node[draw,circle] (a1) at (0,-4.5) {$a$};
					\node[draw,circle] (a2) at (0,0) {$v_1$};
					\node[draw,circle] (a3) at (0,-1.5) {$v_2$};
					\node[draw,circle] (a4) at (0,-3) {$v_3$};
					\node[draw,circle] (b1) at (2.2,0) {$v_5$};
					\node[draw,circle] (b2) at (2.2,-1.5) {$b_1$};
					\node[draw,circle] (b3) at (2.2,-3) {$b_2$};
					\node[draw,circle] (b4) at (2.2,-4.5) {$v_4$};
					\draw (a1) -- (b1);
					\draw (a1) -- (b2);
					\draw (a1) -- (b3);
					\draw (a1) -- (b4);
					\draw (a2) -- (b1);
					\draw (a2) -- (b2);
					\draw (a3) -- (b2);
					\draw (a3) -- (b3);
					\draw (a4) -- (b3);
					\draw (a4) -- (b4);
			\end{tikzpicture}} &
			\parbox{.4\textwidth}{\centering
				\begin{tikzpicture}[scale=.75]
					\node[draw,circle] (a1) at (0,-4.5) {$a$};
					\node[draw,circle] (a2) at (0,0) {$v_1$};
					\node[draw,circle] (a3) at (0,-1.5) {$v_2$};
					\node[draw,circle] (a4) at (0,-3) {$v_3$};
					\node[draw,circle] (b1) at (2,0) {$v_5$};
					\node[draw,circle] (b2) at (2,-1.5) {$b_1$};
					\node[draw,circle] (b3) at (2,-3) {$b_2$};
					\node[draw,circle] (b4) at (2,-4.5) {$v_4$};
					\draw (a1) -- (b1);
					\draw (a1) -- (b2);
					\draw (a1) -- (b3);
					\draw (a1) -- (b4);
					\draw[color=red,line width=1.5pt] (a2) -- (b1);
					\draw (a2) -- (b2);
					\draw (a3) -- (b2);
					\draw (a3) -- (b3);
					\draw (a4) -- (b3);
					\draw[color=red,line width=1.5pt] (a4) -- (b4);
					\draw (a1) -- (a4);
     \draw (a1).. controls (-1.75,-2).. (a2);
					\draw[color=red,line width=1.5pt] (a2) -- (a3);
					\draw[color=red,line width=1.5pt] (a3) -- (a4);
					\draw (a1).. controls (-1.00,-2.5).. (a3);

					\draw (b1) -- (b2);
					\draw (b2) -- (b3);
					\draw (b3) -- (b4);
					\draw (b1).. controls (2.75,-2).. (b3);
					\draw[color=red,line width=1.5pt] (b1).. controls (3.5,-2.75).. (b4);
					\draw (b2).. controls (2.75,-3).. (b4);
			\end{tikzpicture}}\\
			\vspace{.1in}
			\parbox{.4\textwidth}{\centering $G$} &  \vspace{.1in} \parbox{.4\textwidth}{\centering $G^2$}
		\end{tabular}
		\caption{A convex bipartite graph, whose square is not perfect}
		\label{fig:not_perfect}
	\end{figure}

	\subsection{Preliminaries} \label{sec:prelim}
	Given a graph $G = (V,E)$, a set of vertices $S \subseteq V(G)$ is said to be a \emph{clique} in $G$ if they are pairwise adjacent in $G$. On the other hand, for a set of vertices $S \subseteq V(G)$, if no two vertices in $S$ are adjacent in $G$, then $S$ is said to be an \emph{independent set} in $G$. The square of a graph $G=(V,E)$, denoted as $G^2$, is defined to be the graph with $V(G^2)=V(G)$ and $E(G^2)=\{uv: uv\in E(G)$ or there exists a $w\in V(G)$ such that $uw,wv\in E(G)\}$. A graph $H$ is called an \emph{induced subgraph} of a graph $G$ if $V (H) \subseteq V (G)$ and
	$E(H) = \{uv \in E(G) : u, v \in V (H)\}$; $H$ is then also said to be the subgraph induced by
	$V (H)$. A \emph{cycle} in a graph $G$ is defined as a sequence of disjoint vertices, denoted as $(v_1,\ldots, v_k)$ such that $\{ v_iv_{i+1}: 1 \leq i \leq k-1 \} \cup \{ v_kv_i \} \subseteq E(G)$.  The length of a shortest cycle contained in a graph $G$ is called the \emph{girth} of $G$. Let $C=(v_0,v_1,\ldots,v_{k-1})$ be a cycle in $G$. We say that $P$ is a \emph{sub-path} of $C$ if there exist $i,j\in \{0,1,2,\ldots,k-1\}$  such that $P=(v_i,v_{i+1},\ldots,v_j)$ (indices modulo $k$) and $v_iv_j\notin E(C)$.  
	A \emph{hole} is defined as an induced cycle of length at least four. The complement of a hole is called an \emph{antihole}. A graph $G$ is said to be \emph{$H$-free} if $G$ does not contain $H$ as an induced subgraph. A graph is said to be \emph{chordal} if it is hole-free. Let $H$ be a graph and $\mathcal{C}$ be a graph class. By \emph{$H$-free squares of $\mathcal{C}$}, we mean the graphs $G^2$ such that $G^2$ is $H$-free and $G\in \mathcal{C}$.
	
	\medskip
	A graph $G$ is called a \emph{bipartite graph} if its vertices can be partitioned into two disjoint independent sets $A$ and $B$ such that every edge in $G$ has one endpoint in $A$ and the other endpoint in $B$. A bipartite graph with partite sets $A$ and $B$ is usually denoted as \emph{$G=(A,B,E)$}. A bipartite graph is said to be \emph{chordal bipartite} if it does not contain any induced cycle of length at least six. Note that chordal bipartite graphs form a superclass of \emph{convex bipartite graphs} (see Definition~\ref{def:convex}). A graph $G$ is an \emph{interval graph} if the vertices of $G$ can be represented by intervals on the real line such that two vertices in $G$ are adjacent if and only if the corresponding intervals intersect. The corresponding set of intervals is called an {\em interval representation} of $G$. A \emph{proper interval graph} is an interval graph that has an interval representation in which no interval properly contains another. A graph $G$ is said to be \emph{perfect} if for each induced subgraph $H$ of $G$, the chromatic number $\chi(H)=\omega(H)$, the clique number. The well-known \emph{Strong Perfect Graph Theorem} states that \emph{a graph $G$ is perfect if and only if $G$ is both (odd hole)-free and (odd antihole)-free}~\cite{chudnovsky2006strong}. Many graph families, including bipartite graphs, chordal graphs, etc., are perfect.

	\medskip \noindent \textbf{Notation:} For a vertex $v$ in a graph $G$, we denote by $N_{G}(v)$, the set of vertices adjacent to $v$ in $G$ and $N_{G}[v] = N_{G}(v) \cup \{v\}$. 
 For a set $S\subseteq V(G)$, we denote by $G[S]$, the subgraph induced by $S$. For vertices $u,v\in V(G)$, we denote by $d_G(u,v)$, the distance (length of the shortest path) between $u$ and $v$ in $G$. For a graph $G$ and $S \subseteq V(G)$, we use $G-S$ to denote the graph $G[V(G)-S]$. For two graphs $G$ and $H$, we denote by $G\cong H$ if $G$ and $H$ are isomorphic to each other. Depending upon the context, sometimes we denote a cycle (path) of length $k$ as $(v_1,\ldots, v_k)$. 
	


	\section{Observations on the asymptotic bounds}
	
	In this section, we derive some implications (Observations~\ref{obs:general},~\ref{obs:general_lower},~\ref{obs:bipartite}) of the following upper and lower bounds for the chromatic numbers of squares of graphs in the work of Alon and Mohar~\cite{Alon2002} stated below. 
	
	\begin{theorem}[\cite{Alon2002}]\label{thm:alon}
		\begin{myenumerate}
			\item\label{upperbound} Let $\Delta \geq 2$ and $g\geq 7$ be any integers. There exists an absolute constant $c_1$ such that for any graph $G$ with maximum degree $\Delta$ and girth $g$  we have, $\chi(G^2)\leq c_1\frac{\Delta^2}{\log \Delta}$.
			\item \label{lowerbound}  Let $\Delta\geq 2$ and $g\geq 7$ be any integers. There exist an absolute constant $c_2$ and a graph $G$ with maximum degree $\Delta$, girth $g$, and $\chi(G^2)\geq c_2\frac{\Delta^2}{\log \Delta}$
		\end{myenumerate}
	\end{theorem}
	
	Since we are interested in bounding $\chi(G^2)$ in terms of $\omega (G^2)$ (unlike in Theorem~\ref{thm:alon} where the bounds are in terms of the maximum degree, $\Delta$ of $G$), for graphs having girth, $g\geq 7$,  in the following lemma we find a relation between the parameters maximum degree and clique number of $G^2$. 
	
	\begin{lemma}\label{lem:omega} 
		Let $G$ be a graph with girth $g\geq 7$ and maximum degree $\Delta\geq 2$. Then $\omega=\omega(G^2)=\Delta+1$.
	\end{lemma}
	\begin{proof}
		Let $C$ be any maximum clique in $G^2$. Clearly, $|C|\geq \Delta+1$, as for any vertex $v$ in $G$, the set $N_G[v]$ form a clique in $G^2$. Our goal is to prove that $|C|\leq \Delta+1$ as well, which then proves the lemma. Note that if there exists a vertex $v$ in $G$ such that $C\subseteq N_G[v]$, then we are done. Suppose that no such vertex exists in $G$. i.e. $C\nsubseteq N_G[v]$ for any vertex $v$ in $G$. Now, consider a vertex $a\in C$. Then, by our assumption, there exists a vertex $b\in C$ such that $b\notin N_G[a]$. Since $ab\in E(G^2)$ (as $a,b\in C$), we have that there exists a vertex $c\in V(G)$ such that $ac,bc\in E(G)$. As we also have $C\nsubseteq N_G[c]$, there exists a vertex $x\in C$ such that $x\notin N_G[c]$. This implies that $x\neq a,b,c$. Since $ax,bx\in E(G^2)$ (as $a,b,x\in C$), at least one of the following conditions should hold:
		\begin{enumerate}
			\item $ax,bx\in E(G)$.
			\item $bx\in E(G)$, but $ax\notin E(G)$ and there exists a vertex $z$ in $G$ such that $az,zx\in E(G)$. 
			\item $ax\in E(G)$, but $bx\notin E(G)$ and there exists a vertex $y$ in $G$ such that $by,yx\in E(G)$.
		\sloppy	\item $ax,bx\notin E(G)$, and there exist two vertices $y,z$ (possibly, $y=z$) in $G$ such that $by,yx,az,zx\in E(G)$.
		\end{enumerate}
		
		If (1) holds, we have a 4-cycle, namely, $(b,x,a,c)$ in $G$. Note that as $cx\notin E(G)$, we have $c\neq y,z$. If either (2) or (3) hold, we have 5-cycles, namely, $(a,z,x,b,c)$ or $(b,y,x,a,c)$, respectively, in $G$ (note that since $ab\notin E(G)$, we have $z\neq b$ and $y\neq a$). If (4) holds, then we have a 6-cycle, namely, $(b,y,x,z,a,c)$ in $G$ when $y\neq z$, and a 4-cycle in $G$, namely $(b,y=z,a,c)$ when $y=z$. Therefore, in any case, we have a contradiction to the fact that the girth, $g\geq 7$. Thus, we can conclude that our assumption is not true. Hence the lemma.
	\end{proof}
	
	
	Now Observation~\ref{obs:general} (respectively, Observation~\ref{obs:general_lower}) follows from Theorem~\ref{thm:alon}\ref{upperbound} (respectively, Theorem~\ref{thm:alon}\ref{lowerbound}) and Lemma~\ref{lem:omega}.
	
	\begin{observation}\label{obs:general}
		Let $\omega \geq 3$ and $g\geq 7$ be any integers. There exists an absolute constant $c_1$ such that for any graph $G$ with girth $g$ and $\omega=\omega(G^2)$, we have  $\chi(G^2)\leq c_1\frac{(\omega -1)^2}{\log(\omega-1)}$.    
	\end{observation}
	
	\begin{observation}
		\label{obs:general_lower}
		Let $\omega \geq 3$ and $g\geq 7$ be any integers. There exist an absolute constant $c_2$ and a graph $G$ such that $G$ has girth $g$, $\omega=\omega(G^2)$, and $\chi(G^2)\geq c_2\frac{(\omega -1)^2}{\log(\omega-1)}$.
	\end{observation}
	
	Interestingly, the tight example graph in Observation~\ref{obs:general_lower} can be converted to a bipartite graph by using the following simple reduction. This further helps us to prove the existence of bipartite graphs $B$ such that the chromatic number of $B^2$ is $\Omega(\frac{\omega^2}{\log\omega})$, where $\omega$ is the clique number of $B^2$.
	
	\medskip
	\noindent\textbf{Reduction:} Given a graph $G$, we define a bipartite graph $B_G$ (which is obtained by splitting each vertex of $G$) as follows:
	\begin{align*}
		V(B_G)&=A\cup B, \text{ where } A=\{u':u\in V(G\} \text{ and } B=\{u'':u\in V(G)\} \\
		E(B_G)&=\{u'v'': u'\in A, v''\in B, \text{ such that either } u=v \text{ or } uv\in E(G)\}
	\end{align*}
	
	We have the following lemma for the bipartite graph $B_G$, constructed by the above reduction.
	
	\begin{lemma}\label{claim: isomorphic}
		For any graph $G$, we have $B_G^2[A]\cong G^2$ and $B_G^2[B]\cong G^2$.   
	\end{lemma}
	\begin{proof}
		We only give the proof for $B_G^2[A]\cong G^2$, as similar arguments can be used to prove $B_G^2[B]\cong G^2$. By the definition of $B_G$, every vertex $u\in V(G^2)$ corresponds to a vertex $u'\in V(B_G^2[A])$ (this provides us a bijective mapping from $V(G^2)$ to $V(B_G^2[A])$). It is now enough to show that for any two vertices $u,v\in V(G^2)$, we have $uv\in E(G^2)$ if and only if $u'v'\in E(B_G^2[A])$. Suppose that $uv\in E(G^2)$. If $uv\in E(G)$, we then have $u'v'',v''v'\in E(B_G)$, and therefore, $u'v'\in E(B_G^2[A])$. If $uv\notin E(G)$, this implies that there exists a vertex $w\in V(G)$ such that $uw,wv\in E(G)$. This further implies that, $u'w'',w''v'\in E(B_G)$, and therefore $u'v'\in E(B_G^2[A])$. This proves the if part. To prove the converse, assume that $u'v'\in E(B_G^2[A])$. Since $B_G$ is a bipartite graph and both the vertices $u',v'\in A$, this implies that there exists a vertex $w''\in B$ such that $u'w'',w''v'\in E(B_G)$. If either $w''=u''$ or $w''=v''$, by the definition of $B_G$, we then have $uv\in E(G)\subseteq E(G^2)$. Therefore, we can assume that $w''\neq u'',v''$. Then the fact that $u'w'',w''v'\in E(B_G)$ implies that $uw,wv\in E(G)$. This further implies that $uv\in E(G^2)$, and we are done. 
	\end{proof}
	
	We note the following observation.
	\begin{observation}\label{obs:clique_BG}
		For any graph $G$, we have $\omega(G^2)\leq \omega(B_G^2)\leq 2\omega(G^2)$ and $\chi(G^2)\leq \chi(B_G^2)\leq 2\chi(G^2)$.
	\end{observation}
	
	\begin{proof}
		Clearly, $\omega(G^2)\leq \omega(B_G^2)$ and $\chi(G^2)\leq \chi(B_G^2)$, as $G^2$ is present as an induced subgraph in $B_G^2$ by Lemma~\ref{claim: isomorphic}. Since we have $B_G^2[A]\cong G^2$, $B_G^2[B]\cong G^2$ (again, by Lemma~\ref{claim: isomorphic}), and $V(B_G^2)=V(B_G^2[A])\cup V(B_G^2[B])$, we have the remaining inequalities, $\omega(B_G^2)\leq \omega(B_G^2[A]) + \omega(B_G^2[B]) = 2\omega(G^2)$ and $\chi(B_G^2)\leq \chi(B_G^2[A]) + \chi(B_G^2[B]) = 2\chi(G^2)$.  
	\end{proof}
	
	We now have Observation~\ref{obs:bipartite}, which proves the existence of bipartite graphs whose squares have a sub-quadratic lower bound for the chromatic number. 
	\begin{observation} 
		\label{obs:bipartite} 
		Let $\omega'\geq 3$. There exists an absolute constant $c$ and a bipartite graph $H$ with $\omega(H^2)\leq \omega'$ and $\chi(H^2)\geq c\frac{(\omega(H^2) -2)^2}{\log(\omega(H^2)-1)}$.
	\end{observation}
	\begin{proof} 
 Let $g\geq 7$ and $\omega=\frac{\omega'}{2}$. By Observation~\ref{obs:general_lower}, there exists an absolute constant $c_2$ and a graph $G$, with girth $g$ and $\omega=\omega(G^2)$ such that $\chi(G^2)\geq c_2\frac{(\omega -1)^2}{\log(\omega-1)}$. Let $H=B_G$ be the bipartite graph obtained from $G$ by applying the above reduction. Then by Observation~\ref{obs:clique_BG}, we have that  $\omega\leq \omega(H^2)\leq 2\omega=\omega'$. Again, by Observation~\ref{obs:clique_BG} and from the above inequalities, by setting $c=c_2/4$, we have the following:
		\begin{align*}
			\chi(H^2)\geq\chi(G^2)&\geq  c_2\frac{(\omega -1)^2}{\log(\omega-1)}\\
			& \geq c_2\frac{(\frac{\omega(H^2)}{2} -1)^2}{\log(\omega(H^2)-1)}\\
			& \geq c\frac{(\omega(H^2) -2)^2}{\log(\omega(H^2)-1)} 
		\end{align*}
	\end{proof}
	
	\section{Proof of Theorem~\ref{thm:theorem-convex-bipartite}}\label{sec:convex_main}

	\medskip
	
	Let $G=(A,B,E)$ be a convex bipartite graph. Then by Definition~\ref{def:convex}, we have an ordering, say $<_B$ for the vertices in $B$ such that, for each vertex $a\in A$, the vertices in $N_G(a)$ appear consecutively with respect to $<_B$. 
	Now, consider the vertices in $B$ (with respect to the ordering $<_B$) as points on the real line. Then, for each vertex $a\in A$, 
	we define an interval $I_a$ with $l(I_a)=\min_{<_B}\{b:b\in N_G(a)\}$ and $r(I_a)=\max_{<_B}\{b:b\in N_G(a)\}$. 
	Note that, for any two vertices $a,a'\in A$, $aa'\in E(G^2[A])$ if and only if $I_a\cap I_{a'}\neq \emptyset$. Hence, the collection of intervals, $\{I_a\}_{a\in A}$ is an interval representation of $G^2[A]$. 
	We then have the following observation, which is also noted in~\cite{Le2019hardness}. 
	\begin{observation}[\cite{Le2019hardness}]\label{obs:interval}
		Let $G=(A,B,E)$ be a convex bipartite graph. Then the subgraph $G^2[A]$ is an interval graph.
	\end{observation}
	Moreover, it is also known that $G^2[B]$ is an interval graph~\cite{Le2019hardness}.
	But in the following lemma,  we prove a stronger observation for $G^2[B]$, which is useful in proving some of our results. Note that for a graph $G$, an ordering $<$ of $V(G)$ is called a \emph{proper vertex ordering} if for any three vertices $u,v,w$ in $G$ such that $u<v<w$, $uw\in E(G)$ implies that $uv,vw\in E(G)$. It is a well-known fact that the proper interval graphs are exactly those graphs whose vertex set admits a proper vertex ordering. 
	\begin{lemma}\label{obs:properinterval}
		Let $G=(A,B,E)$ be a convex bipartite graph. Then, the ordering $<_B$ of the vertices in $B$ is a proper vertex ordering of $G^2[B]$. Consequently, the subgraph $G^2[B]$ is a proper interval graph.
	\end{lemma}
	\begin{proof}
		Consider any three vertices, say $u,v,w\in B$, such that $u<_B v <_B w$. Suppose that $uw\in E(G^2)$. This implies that there exists a vertex $a\in A$ such that $au,aw\in E(G)$. By the definition of $<_B$, $N_G(a)$ is consecutive. Thus we have $av\in E(G)$. As $au,aw\in E(G)$, we can therefore conclude that $uv,vw\in E(G^2)$. This proves that $<_B$ is a proper vertex ordering of vertices in $G^2[B]$. Hence, the observation.
	\end{proof}

	Now we define an ordering $<_A$ for the vertices in $A$ as follows: for any pair of vertices, say $a_i,a_l\in A$, we say that $a_i<_A a_l$ if and only if $r(I_{a_i})\leq r(I_{a_l})$ (where $\{I_a\}_{a\in A}$ is the same interval representation of $G^2[A]$ defined earlier, and if $r(I_{a_i})=r(I_{a_l})$ we can have either $a_i<_A a_l$ or $a_l<_A a_i$). In the remainder of the section, for a convex bipartite graph $G=(A,B,E)$, we assume that the vertices of the sets $A$ and $B$ follow the orderings $<_A$ and $<_B$ respectively. 
	If $|A|=m$ and $|B|=n$, we then denote $A=\{a_1,a_2,\ldots,a_m\}$, where $a_1<_A\cdots<_A a_m$ and $B=\{b_1,b_2,\ldots,b_n\}$, where $b_1<_B\cdots<_B b_n$. Throughout the section, we denote by $\omega$, the size of the maximum clique of $G^2$. For vertices $x,y\in A$ (respectively, $x,y\in B$), the notation $x\leq_A y$ (respectively, $x\leq_B y$) includes the possibility that $x=y$. We now infer the following observations.
 \begin{observation}\label{obs:bipsquare}
    Let $G=(A,B,E)$ be a bipartite graph. Let $a\in A$ and $b\in B$ be such that $ab\in E(G^2)$ then $ab\in E(G)$.
 \end{observation}
 \begin{proof}
    Note that for $a\in A$ and $b\in B$, $d_G(a,b)\neq 2$. Therefore, by the definition of $G^2$, $ab\in E(G^2)$ implies that $ab\in E(G)$.
 \end{proof}
 
 \begin{observation}\label{obs:orderprop}
     Let $a,a'\in A$ be such that $a<_A a'$, and $b,b'\in B$ be such that $b'<_B b$. If $b\in N_G(a)$ and $b'\in N_G(a')$ then $b\in N_G(a')$.
 \end{observation}
 \begin{proof}
     Since $a<_A a'$, we have $r(I_{a})\leq r(I_{a'})$. Now as $b\in N_G(a)$, $b'\in N_G(a')$, $b'<_B b$, and the vertices in $N_G(a')$ appear consecutive with respect to the ordering $<_B$, we have $b\in N_G(a')$.
 \end{proof}

	\begin{observation}\label{obs:clique}
	\sloppy	For each $i\in \{1,2,\ldots,m\}$ and $j\in \{1,2,\ldots,n\}$, the sets $N_{G^2}(a_i)\cap \{a_{i+1},\ldots,a_m\}$ and $N_{G^2}(b_j)\cap \{b_{j+1},\ldots,b_n\}$ are both cliques in $G^2$.
	\end{observation}
	\begin{proof}
		Let  $b_k$ be the largest indexed vertex in $N_{G^2}(b_j)\cap \{b_{j+1},\ldots,b_n\}$. Since $b_jb_k\in E(G^2)$, there exists a vertex $a\in A$, such that $ab_j,ab_k \in E(G)$. Then, as vertices in $N_G(a)$ appear consecutively with respect to the ordering $<_B$, we have that $ab_l\in E(G)$, for each  $l$ such that $j \leq l \leq k$. This implies that for any $l,l'$ such that $j \leq l,l' \leq k$, we have $b_lb_{l'}\in E(G^2)$. Therefore, $N_{G^2}(b_j)\cap \{b_{j+1},\ldots,b_n\}$ is a clique in $G^2$.
		
		\smallskip
		
		Now suppose that $N_{G^2}(a_i)\cap \{a_{i+1},\ldots,a_m\}$ is not a clique in $G^2$. Then there exist vertices, $a_k,a_l\in N_{G^2}(a_i)\cap \{a_{i+1},\ldots,a_m\}$ such that $a_ka_l\notin E(G^2)$. Without loss of generality, we can assume that $a_k<_A a_l$. Then, by the definition of $<_A$, we have that $r(I_{a_k})\leq r(I_{a_l})$. Recall that $\{I_a\}_{a\in A}$ is an interval representation of $G^2[A]$. Since $a_ka_l\notin E(G^2)$ and $r(I_{a_k})\leq r(I_{a_l})$, we can infer that $r(I_{a_k})<l(I_{a_l})$. As $a_i<_A a_k$, we then have  $r(I_{a_i})\leq r(I_{a_k})<l(I_{a_l})$. This further implies that $a_ia_l\notin E(G^2)$. This contradicts the fact that $a_l\in N_{G^2}(a_i)$. Hence, the observation.
	\end{proof}
	
	To prove the upper bound in Theorem~\ref{thm:theorem-convex-bipartite}, our idea is to recursively construct subgraphs of $G^2$, namely, $H_j$, for each integer $j$ down from $n$ to $1$, and show that the chromatic number of each of these subgraphs is bounded by $\left\lfloor\frac{3\omega}{2}\right\rfloor$ (Lemma~\ref{lem:main-convex-bipartite}). 
	
	\medskip
	\noindent\textbf{Subgraphs $H_j$:}
	For each integer $j$ down from $n$ to $1$, we define an induced subgraph $H_j=G^2[A\cup \{b_j,b_{j+1},\ldots,b_n\}]$, where the vertices in $A$ follows the ordering $<_A$, and $b_j<_B b_{j+1}<_B\cdots<_B b_n$. Clearly, $H_1=G^2$, and $G^2[A]$ is an induced subgraph of $H_j$ for each $j\in \{1,2,\ldots,n\}$. Now, to prove the upper bound in Theorem~\ref{thm:theorem-convex-bipartite}, it is enough to prove the following lemma.
	\begin{lemma}\label{lem:main-convex-bipartite}
		For each $j\in \{1,2,\ldots,n\}$, we have $\chi(H_j)\leq \lfloor \frac{3\omega}{2}\rfloor$. 
	\end{lemma}
	
	\subsection{Some observations on subgraphs $H_j$}\label{subsec:main}
	
	First, the following observation is immediate from the definition of $H_j$.
	
	\begin{observation} \label{obs:endpoint}
		Let $xy \in E(H_j) \setminus E(H_{j+1})$ for some $j\in \{1,2,\ldots,n-1\}$. Then either $x=b_j$ or $y=b_j$.
	\end{observation}
	
	We have the following observation due to the definitions of $<_A$ and $<_B$.
	\begin{observation} \label{obs:chain}
		For $j\in \{1,2,\ldots,n\}$, let $N_{H_j}(b_j)\cap A = \{a_{i_1},a_{i_2},\ldots,a_{i_k}\}$, where  $a_{i_1}<_A a_{i_2}<_A\cdots<_A a_{i_k}$. Then $N_{H_j}(a_{i_1})\cap B\subseteq N_{H_j}(a_{i_2})\cap B\subseteq \ldots \subseteq N_{H_j}(a_{i_k})\cap B$.
	\end{observation}
	
	We also note the following.
	\begin{observation} \label{obs:A_jB_j}
		For $j\in \{1,2,\ldots,n\}$, let $A_j=N_{H_j}(b_j)\cap A$ and $B_j= N_{H_j}(b_j)\cap B$. Then $|A_j|\leq \omega-1$ and $|B_j|\leq \omega-2$.
	\end{observation}
	\begin{proof}
		Let $A_j=N_{H_j}(b_j)\cap A$. It is easy to see that $A_j\cup \{b_j\}$ is a clique in $G^2$. This implies that $|A_j|\leq \omega-1$. Let $B_j= N_{H_j}(b_j)\cap B=\{b_{j_1}, b_{j_2},\cdots,b_{j_l}\}$, where $b_{j_1}<_B b_{j_2}<_B\cdots<_B b_{j_l}$. By Observation~\ref{obs:clique}, we have that $B_j$ is a clique in $G^2$. Therefore, $B_j\cup\{b_j\}$ is also a clique in $G^2$. Now, consider the vertex $b_{j_l}\in B_j$. As $b_jb_{j_l}\in E(G^2[B])$, there exists a vertex $a\in A_j$ such that $ab_j,ab_{j_l}\in E(G)\subseteq E(G^2)$. Then, by the definition of $<_B$, we have that neighbors of $a$ are consecutive in $G$. Therefore, we have $B_j \cup \{b_j\} \subseteq N_{G}(a)\subseteq N_{G^2}(a)$. This implies that $B_j \cup \{b_j,a\}$ is a clique in $G^2$. This further implies that $|B_j|\leq \omega-2$.
	\end{proof}
	
	\subsection{Proof of Lemma~\ref{lem:main-convex-bipartite}}
	\noindent\textbf{Comment:} We encourage the reader to go through the following proof jointly with the algorithm described in Section~\ref{sec:algo} for a better understanding.
	\begin{proof}
        The proof is based on reverse induction on the index $i$ of the subgraphs $H_i$, $i\in \{1,2,\ldots,n\}$ defined above. For 
        each $i\in \{1,2,\ldots,n\}$, let $$A_i=N_{H_i}(b_i)\cap A$$ $$B_i= N_{H_i}(b_i)\cap B$$
		Recall that $G^2[A]$ is an interval graph (by Observation~\ref{obs:interval}), and the size of the maximum clique in $G^2[A]$ is at most $\omega$, where $\omega=\omega(G^2)$. Therefore, $G^2[A]$ is $\omega$-colorable. Consider the base case, $i=n$. Note that $H_n=G^2[A\cup \{b_n\}]$ and 
		$B_n=\emptyset$. Also, by Observation~\ref{obs:A_jB_j}, we have that $|A_n|\leq \omega-1$. Therefore, as $G^2[A]$ is $\omega$-colorable, and $|N_{H_n}(b_n)|=|A_n|\leq \omega-1$, we can easily extend any $\omega$-coloring of $G^2[A]$ to an $\omega$-coloring of $H_n$, and we are done. Now, assume to the induction hypothesis that $\chi(H_j)\leq \lfloor\frac{3\omega}{2}\rfloor$ for any $j>i$. Consider $H_i=G^2[A\cup \{b_n,b_{n-1},\ldots,b_i\}]$. Our goal is to prove that $\chi(H_i)\leq \lfloor\frac{3\omega}{2}\rfloor$.
		
		\begin{definition}[Special coloring]\label{def:special}
			Any proper coloring of $H_{i+1}$ that uses at most $\lfloor\frac{3\omega}{2}\rfloor$ colors is a special coloring of $H_{i+1}$.
		\end{definition}
		
		Due to our induction hypothesis, a special coloring of $H_{i+1}$ always exists. For a special coloring $c$ of $H_{i+1}$, if there exists a color, say $x$, which is not used to color any vertex in $A_i\cup B_i$, then we can extend the coloring $c$ to $H_i$ by assigning the color $x$ to $b_i$. This would have proved the lemma.
		
		\smallskip
		
		Otherwise, if a special coloring $c$ of $H_{i+1}$ is not extendable, i.e. all the $\lfloor\frac{3\omega}{2}\rfloor$ colors have been used up by the vertices $A_i \cup B_i$, our intention is to convert $c$ to another special coloring of $H_{i+1}$ by \emph{Kempe changes}. i.e. \emph{swapping} the colors of some vertices. Before that, we introduce some definitions and prove some properties for \emph{``non-extendable special colorings''} of $H_{i+1}$.
		
		\begin{definition}[Non-extendable special coloring]\label{def:extend}
			A special coloring $c$ of $H_{i+1}$ is ``\emph{non-extendable}'' if all the $\lfloor\frac{3\omega}{2}\rfloor$ colors have been used up by the vertices in $A_i \cup B_i$. Otherwise, $c$ is an extendable special coloring. 
		\end{definition}
		
		\begin{myclaim}\label{claim:unique}
			\textit{Let $c$ be a non-extendable special coloring. Then there exist $\lfloor\frac{\omega}{2}\rfloor+2$ vertices in $A_i$ whose color (with respect to $c$) is not assigned to any other vertex in $A_i\cup B_i$.}
		\end{myclaim}
		\begin{proof}
			Note that both the sets $A_i$ and $B_i$ are cliques in $H_i$ (due to the fact $A_i=N_G(b_i)\cap A$, and Observation~\ref{obs:clique}). By Observation~\ref{obs:A_jB_j}, we have $|B_i|\leq \omega-2$. Since $c$ uses $\lfloor\frac{3\omega}{2}\rfloor$ colors in $A_i\cup B_i$, we then have that the vertices in $A_i$ use at least $\lfloor\frac{\omega}{2}\rfloor+2$ extra colors which are not used for coloring any of the vertices in $B_i$. This proves the claim (since $A_i$ is a clique, no two vertices in $A_i$ can have the same color).
   \renewcommand{\qedsymbol}{$\blacksquare$}
		\end{proof}

		\begin{definition}\label{def:a(c)}
			For a non-extendable special coloring $c$ of $H_{i+1}$, let $a_{c}=\min_{<A} \{a\in A_i:$ the color $c(a)$ is not used by any other vertex in $A_i\cup B_i\}$. We call $a_{c}$ the pivot vertex with respect to $c$.
		\end{definition}
		
		Note that the pivot vertex $a_{c}$ is well defined for any non-extendable special coloring $c$ of $H_{i+1}$, by Claim~\ref{claim:unique}. Now Claim~\ref{claim:A_i'} is immediate from Claim~\ref{claim:unique} and the definition of $a_{c}$.
		
		\begin{myclaim}\label{claim:A_i'}
			\textit{For a non-extendable special coloring $c$ of $H_{i+1}$, let $A_i'=\{a\in A_i: a_{c} \leq_A a$\}. Then $|A_i'| \geq \Big\lfloor\frac{\omega}{2}\Big\rfloor+2$.}
		\end{myclaim}
		
		\begin{myclaim}\label{claim:pivot}  
			\textit{For a non-extendable special coloring $c$ of $H_{i+1}$, one of the following holds:
				\begin{myenumerate}
					\item\label{it:a} There are $\Big\lfloor\frac{3\omega}{2}\Big\rfloor$ vertices in $N_{H_{i+1}}[a_{c}]$, all receiving distinct colors with respect to $c$, or
					\item \label{it:b} There exists an extendable special coloring $c'$ of $H_{i+1}$.
			\end{myenumerate}}
		\end{myclaim}
		\begin{proof}
			Suppose \ref{it:a} is not true, then there exists a color, say $x\in \{1,2,\ldots,\lfloor\frac{3\omega}{2}\rfloor\}$, such that $x\neq c(a_{c})$, and the color $x$ is not used by any neighbor of $a_{c}$. Then, let $c'$ be a coloring of $H_{i+1}$ such that $c'(v)=c(v)$ for each vertex $v\neq a_{c}$, and $c'(a_{c})=x$. Clearly, $c'$ is a special coloring of $H_{i+1}$ and it is also extendable, since with respect to the coloring $c'$, the vertex $b_i$ can now be assigned with the color $c(a_{c})$ to obtain a $\Big\lfloor\frac{3\omega}{2}\Big\rfloor$-coloring of $H_i$. This implies \ref{it:b} is true, and hence the claim.
   \renewcommand{\qedsymbol}{$\blacksquare$}
		\end{proof}
		
		For a non-extendable special coloring $c$ of $H_{i+1}$, if Claim~\ref{claim:pivot}\ref{it:b} is true, then we are done. Therefore, the difficult case is when Claim~\ref{claim:pivot}\ref{it:a} is true, but not \ref{claim:pivot}\ref{it:b}.
		
		\begin{myclaim}\label{claim:pivotcolor}
			\textit{For a non-extendable special coloring $c$ of $H_{i+1}$, let \ref{claim:pivot}\ref{it:a} is true. Then there exists a color $y \in \Big\{1,2,\ldots,\Big\lfloor\frac{3\omega}{2}\Big\rfloor\Big\}$ such that $y\neq c(a_{c})$, and the color $y$ is not assigned to any vertex of the set, $$ S=(N_{H_{i+1}}(a_{c})\cap B)  \cup \{a\in (N_{H_{i+1}}(a_{c})\cap A): a_{c}<_A a\}$$
				(Note that $S$ is exactly the set obtained from $N_{H_{i+1}}(a_{c})$, by deleting the vertices that lies before $a_{c}$ in the ordering $<_A$)}
		\end{myclaim}
		\begin{proof}
			Let $S_1= N_{H_{i+1}}(a_{c})\cap B$ and $S_2=\{a\in (N_{H_{i+1}}(a_{c})\cap A): a_{c}<_A a\}$. Then $S=S_1\cup S_2$. Let $A_i'=\{a\in A_i: a_{c} \leq_A a$\} 
			(as in the statement of Claim~\ref{claim:A_i'}). Then by Observation~\ref{obs:chain}, we have that $S_1\subseteq (N_{H_i}(a)\cap B)$ for each vertex $a\in A_i'$. This implies that $A_i'\cup S_1\cup\{b_i\}$  form a clique in $H_i$. This further implies that $|A_i'\cup S_1 \cup \{b_i\}|\leq \omega$. Since $|A_i'| \geq \Big\lfloor\frac{\omega}{2}\Big\rfloor+2$ (by Claim~\ref{claim:A_i'}), we then have $|S_1\cup \{bi\}|\leq \lceil \frac{\omega}{2}\rceil -2$. This implies that $|S_1|\leq \lceil \frac{\omega}{2}\rceil -3$. 
            Note that $\lfloor\frac{3\omega}{2}\rfloor-1$ colors are used in $N_{H_{i+1}}(a_{c})$ (since \ref{claim:pivot}\ref{it:a} is true). i.e. the set $N_{H_{i+1}}(a_{c})=(N_{H_{i+1}}(a_{c})\cap A)\cup S_1$ uses $\lfloor\frac{3\omega}{2}\rfloor-1$ colors. Therefore, as $|S_1|\leq \lceil \frac{\omega}{2}\rceil -3$, we can conclude that the set $N_{H_{i+1}}(a_{c})\cap A$ uses at least $\omega+1$ colors that is not used in $S_1$. Now as the set $\{a_{c}\}\cup S_2$ form a clique (by Observation~\ref{obs:clique}), we have that the set $S_2$ can use only at most $\omega -1$ colors. Therefore, the fact that $N_{H_{i+1}}(a_{c})\cap A$ uses at least $\omega+1$ colors implies that there exists a vertex, say $a\in (N_{H_{i+1}}(a_{c})\cap A)$ such that $a<_A a_{c}$ and $y=c(a)\neq c(v)$ for any vertex $v\in S_1\cup S_2\cup \{a_{c}\}$. This proves the claim.
            \renewcommand{\qedsymbol}{$\blacksquare$}
		\end{proof}
		
		\medskip
		
		Let $c$ be a non-extendable special coloring of $H_{i+1}$ for which \ref{claim:pivot}\ref{it:a} is true. Then by Claim~\ref{claim:pivotcolor}, there exists a color $y \in \{1,2,\ldots,\lfloor\frac{3\omega}{2}\rfloor\}$ such that $y\neq c(a_{c})$, and the color $y$ is not assigned to any vertex of the set $$ S=(N_{H_{i+1}}(a_{c})\cap B)  \cup \{a\in (N_{H_{i+1}}(a_{c})\cap A): a_{c}<_A a\}$$ Let $x=c(a_{c})$. Recall that $a_c$ is called as the \emph{pivot vertex} with respect to $c$. In the rest of the section, with respect to the coloring $c$, we call $x$ the \emph{pivot color}, and $y$ the \emph{partner color}. Define, $X=\{v\in V(H_{i+1}): c(v)=x\}$ and $Y=\{v\in V(H_{i+1}): c(v)=y\}$.  Let $D$ denote the component of the induced subgraph $H_{i+1}[X\cup Y]$ that contains the pivot vertex $a_{c}$. We call $D$ the \emph{Kempe component (defined by the pivot and partner colors) containing the pivot} with respect to the coloring $c$.  In the claim below, we evaluate the properties of this Kempe Component. In particular, we intend to prove that $D$  is completely contained in the set $A$. Moreover, we will see that the vertices in the Kempe component containing the pivot vertex and partners are arranged in $A$ in a specific order. 
		
		\begin{myclaim}\label{claim:V_j}
			\textit{Let $c$ be a non-extendable special coloring of $H_{i+1}$, for which \ref{claim:pivot}\ref{it:a} is true and let $D$ be the Kempe component (defined by the pivot and partner colors) containing the pivot with respect to $c$. For each $l\geq 0$, let $V_l=\{v\in V(D): d_{H_i}(v,a_{c} )=l\}$, where $V_0=\{ a_{c} \}$. Then, the following conditions hold.
				\begin{myenumerate}
					\item \label{claim:subset_A} For each $l\geq 0$, we have $V_l\subseteq A$. Moreover, for any pair of vertices $v\in V_{l+1}$, $u\in V_{l}$ such that $uv\in E(H_{i+1})$, we have $v<_A u$. 
					\item \label{claim:neighbor_B} For each $l\geq 2$, we have $N(b)\cap V_l=\emptyset$ for any $b\in (V(H_i)\cap B)$, .
			\end{myenumerate}}
		\end{myclaim}
		
		\begin{proof}
			For the sake of contradiction, assume that the claim is not true. Let $k=\min\{l:l\geq0\}$, such that our claim is not true (i.e. for any $l<k$, the Part~\ref{claim:subset_A} holds if $l\geq 0$, and the Part~\ref{claim:neighbor_B} holds if $l\geq 2$).
			
			\medskip
			
			Note that $k\neq 1$, since $V_0$ has only the pivot vertex $a_c$ in it, and by the definition of partner color $y$, we have $V_1\subseteq \{a\in (N_{H_{i+1}}(a_{c})\cap A): a<_A a_c\}$.
			
			\medskip
			
			Suppose that $k=2$.
			
			\medskip
			
			\noindent\textbf{Part~\ref{claim:subset_A} for $k=2$:} Let $v\in V_2$. Then there exists a vertex $u\in V_1$ such that $uv\in E(H_{i+1})$, $c(u)=y$, and $c(v)=c(a_{c})=x$. (Note that the colors on the vertices belonging to the sets $V_l$ and $V_{l+1}$ alternate between $x$ and $y$ for each $l\geq 0$.) Recall that $x=c(a_{c})$ is not assigned to any other vertex in $N_{H_i}(b_i)=A_i\cup B_i$. 
			As $u<_A a_{c}$, $a_{c}\in N_{H_i}(b_i)$, and $b_i$ is the least indexed vertex in $V(H_i)\cap B$, we then have by the definition of $<_A$ that, $(N_{H_i}(u)\cap B)\subseteq (N_{H_i}(a_{c})\cap B)\subseteq B_i$.  Now, since $v\in N_{H_i}(u)$ and $c(v)=x$, we can conclude that $v\notin B$. Therefore, $v\in A$. To show that the Part~\ref{claim:subset_A} is true for $k=2$, now it is enough to prove that $v<_A u$. Assume to the contrary that $u<_A v$. Since $a_{c},v\in N_{H_i}(u)$ and $u<_A a_{c},v$, we then have by Observation~\ref{obs:clique} that $va_{c}\in E(H_{i+1})$. As $c(v)=c(a_{c})=x$, this contradicts the fact that $c$ is a proper coloring of $H_{i+1}$. Thus, we can conclude that $v<_A u<_A a_{c}$. Therefore, Part~\ref{claim:subset_A} is true for $k=2$.
			
			\medskip
			
			\noindent\textbf{Part~\ref{claim:neighbor_B}  for $k=2$:} Recall that $v<_A a_{c}$, $a_{c}\in N_{H_i}(b_i)$, and $b_i$ is the least indexed vertex in $V(H_i)\cap B$. We then have that $(N_{H_i}(v)\cap B)\subseteq N_{H_i}(a_{c})$ (by the definition of $<_A)$. Thus, if $vb\in E(H_i)$ for some $b\in V(H_i)\cap B$, we then have $va_{c}\in E(H_{i})$ and therefore, $va_{c}\in E(H_{i+1})$ by Observation~\ref{obs:endpoint}.  As $c(v)=c(a_{c})=x$, this again contradicts the fact that $c$ is a proper coloring of $H_{i+1}$. Therefore, Part~\ref{claim:neighbor_B} is also true for $k=2$ .
			
			\smallskip
			Thus, we can conclude that our claim is true for $k=2$, and hence we can assume that $k>2$.
			
			\smallskip
			\noindent\textbf{Part~\ref{claim:subset_A}  for $k>2$:} Let $v\in V_k$. Then there exists a vertex $u\in V_{k-1}$, $w\in V_{k-2}$ such that $vu,uw\in E(H_{i+1})$, and $c(w)=c(v)$. Note that $k-1\geq 2$, and therefore, by the minimality of $k$, we have that~\ref{claim:neighbor_B} is true for $k-1$. Since $u\in V_{k-1}$, we then have that $ub\notin E(H_i)$ for any $b\in  V(H_i)\cap B$. As $uv\in E(H_{i+1})\subseteq E(H_i)$, this implies that $v\notin B$. Therefore, we can conclude that $V_k\subseteq A$. Now, suppose that $u<_A v$. Again, by the minimality of $k$, we have that~\ref{claim:subset_A} is true for $k-1$. Since $u\in V_{k-1}$ and $w\in V_{k-2}$ are such that $uw\in E(H_{i+1})$, this implies that $u<_A w$. This further implies that $u<_A v,w$. By Observation~\ref{obs:clique}, we then have that $vw\in E(H_{i+1})$. As $c(v)=c(w)=x$, this contradicts the fact that $c$ is a proper coloring of $H_{i+1}$. Thus, we can conclude that $v<_A u$. This proves that the Part~\ref{claim:subset_A} is true for $k$. 
			
			\medskip
			
			\noindent\textbf{Part~\ref{claim:neighbor_B}  for $k>2$:} Since  $v<_A u$, we also have $r(I_v)\leq r(I_{u})$ (by the definition of $<_A$). Recall that $u\in V_{k-1}$ and Part~\ref{claim:neighbor_B} holds for $k-1$. Therefore, the fact that $ub\notin E(H_i)$ for any $b\in  V(H_i)\cap B$ implies that $vb\notin E(H_i)$ for any $b\in  V(H_i)\cap B$. This proves that Part~\ref{claim:neighbor_B} is also true for $k$.  
			\smallskip
			
			Our claim is true as we have a contradiction for the existence of $k$.
   \renewcommand{\qedsymbol}{$\blacksquare$}
		\end{proof}
		
		\medskip
		Now define a new proper coloring, say $\phi_c$ of $H_{i+1}$ obtained from $c$ by \emph{swapping the pivot color and partner color on the vertices belonging to the Kempe component $D$}. Formally, it can be defined as follows: 
		
		\[
		\phi_c(v)= 
		\begin{cases}
			c(v),&  v\in V(H_{i+1})\setminus D\\
			x,     & v\in D\cap Y \\
			y,     & v\in D\cap X 
		\end{cases}
		\]
		Clearly,  $\phi_c$ is a special coloring of $H_{i+1}$, and for every non-extendable special coloring, $c$ of $H_{i+1}$ for which \ref{claim:pivot}\ref{it:a} is true,  $\phi_c$ exists. If $\phi_c$ is an extendable special coloring for $H_{i+1}$, then we are done. Otherwise, in the following claim, we prove \emph{a strictly decreasing property of the pivot vertex in the modified special coloring $\phi_c$}. As the number of vertices in the graph is finite, this claim guarantees that we will finally end up having a coloring, say $c^*$, for which $\phi_{c^*}$ is an extendable special coloring. 
		
		
		\begin{myclaim}\label{claim:main}
			\textit{Let $c$ be a non-extendable special coloring of $H_{i+1}$ for which \ref{claim:pivot}\ref{it:a} is true. If $\phi_c$ is non-extendable in $H_{i+1}$ then $a_{\phi_c(v)}<_A a_{c}$}.
		\end{myclaim}
		
		\begin{proof}
			Recall that $\phi_c$ is a special coloring of $H_{i+1}$ and $\phi_c(a_c)=y$. 
           If $\phi_c$ is  non-extendable in $H_{i+1}$, we then have that all the $\Big\lfloor\frac{3\omega}{2}\Big\rfloor$ colors are used in the set $A_i\cup B_i$ with respect to the special coloring $\phi_c$. In particular, now there exists a vertex, say $v\in (A_i\cup B_i)\setminus \{a_{c}\}$ that has the color $x$ on it. i.e. $\phi_c(v)=x$ but $c(v)\neq x$. Since we have recolored only the vertices in the Kempe component $D$ (containing the pivot and partners with respect to $c$) to obtain the new coloring $\phi_c$ from $c$, this implies that $v\in D\subseteq A$, $v<_A a_{c}$ (by Claim~\ref{claim:V_j}\ref{claim:subset_A}), and $c(v)=y$. 
           Moreover, with respect to the coloring $\phi_c$, no vertex in $B_i$ has been colored $x$ (since $D\subseteq A$ by Claim~\ref{claim:V_j}\ref{claim:subset_A}). Also, since $A_i$ is a clique, $v$ is the only vertex in $A_i\cup B_i$ that has the color $x$ on it. Therefore, we now have a coloring $\phi_c$ of $H_{i+1}$ with the property that there exists a vertex $v\in A_i$, $v<_A a_{c}$ such that $\phi_c(v)=x$, and the color $x$ is not assigned to any other vertex in $A_i\cup B_i$.  By the definition of the pivot vertex with respect to $\phi_c$, we then have $a_{\phi_c}\leq _A v$. Further, $v<_A a_{c}$ implies that $a_{\phi_c}<_Aa_{c}$. Therefore, our claim is true.   
           \renewcommand{\qedsymbol}{$\blacksquare$}
		\end{proof}
		
		Now, we are ready to conclude the proof of Lemma~\ref{lem:main-convex-bipartite}. If there is an extendable special coloring for $H_{i+1}$, then we are done. Suppose that there does not exist an extendable special coloring for $H_{i+1}$. Let $c^*$ be a non-extendable special coloring of $H_{i+1}$ having the property that, for any non-extendable special coloring $c$ of $H_{i+1}$ different from $c^*$, we have $a_{c^*}\leq_A a_{c}$. Clearly, \ref{claim:pivot}\ref{it:a} is true for $c^*$. But then by Claim~\ref{claim:main}, we have $a_{\phi_{c^*}(v)}<_A a_{c^*}$. This is a contradiction to the choice of $c^*$. This completes the proof of Lemma~\ref{lem:main-convex-bipartite}.
	\end{proof}
	
	The existence of convex bipartite graphs $H$ with $\chi(H^2)\geq \frac{5\omega(H^2)}{4} - 2$ is proved below. This completes the proof of Theorem~\ref{thm:theorem-convex-bipartite}. 
	
	\subsubsection{Lower bound construction} \label{app:lower}
	For the graph $H$ in Figure \ref{fig:Example}, 
	for each $i\in \{1,2,\ldots,5\}$, the set $Q_i$ represents a clique on $q$ vertices. For distinct $i,j\in \{1,2,\ldots,5\}$ an edge connecting the sets $Q_i$ and $Q_j$ indicates the presence of all edges of the form $q_iq_j$, where $q_i\in Q_i$ and $q_j\in Q_j$, and for $k\in \{1,2,3\}$, an edge connecting a vertex $z_k$ and a set $Q_i$ indicates the presence of all edges of the form $z_kq_i$, where $q_i\in Q_i$. The graph $H$ shown in Figure \ref{fig:Example} is a convex bipartite graph because, as shown in Figure~\ref{fig:Example}, the vertices in $B$ with an ordering $<_B$: $Q_2<_B z_2<_B z_3 <_B Q_3$ (where vertices in the sets $Q_2$ and $Q_3$ can be ordered arbitrarily among themselves) satisfies the consecutive property in the definition of a convex bipartite graph. Since the set $B=Q_2\cup \{z_2,z_3\}\cup Q_3$ form a clique in $H^2$, we have $\omega(H^2)=2q + 3$. The structure of the induced subgraph, $H'=H^2[\bigcup_{i=1}^5 Q_i]$ is commonly known as the \textit{``blow-up"} of a 5-cycle. It can be seen that $\chi(H')=\frac{5}{2}q$, and therefore it is not difficult to verify that, $\chi(H^2)=\frac{5}{2}q+2 =\frac{5(\omega(H^2)-3)}{4}+2\geq \frac{5}{4}\omega(H^2)-2$.
	
	\medskip
	\definecolor{myblue}{RGB}{80,80,160}
	\definecolor{mygreen}{RGB}{80,160,80}
	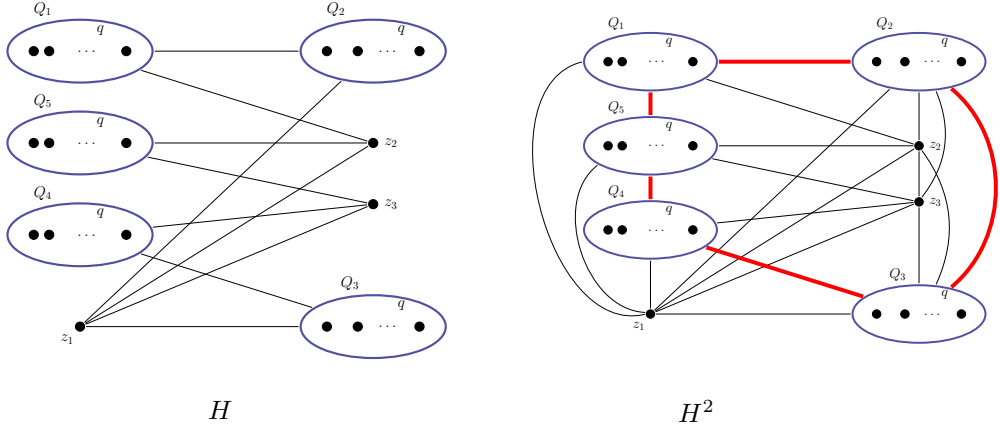
\begin{figure}
         \begin{tabular}{p{.4\textwidth}p{.4\textwidth}}
				\parbox{.4\textwidth}{\centering
					
		\scalebox{.5}{
			\resizebox{.8\textwidth}{!}{
				\begin{tikzpicture}[thick, auto,
					every node/.style={draw,circle,inner sep=3pt},
					fsnode/.style={fill=black},
					ssnode/.style={fill=black},
					every fit/.style={ellipse,draw,inner sep=16pt,text width=2.2cm},shorten >= 1pt,shorten <= 1pt
					]
					
					\node[fill=black, circle] (f1) at (0, 0) {};
					\node[fill=black, circle] (f2) at (0.5, 0) {};
					\node[fill=black, circle] (f6) at (3, 0) {};
					\path (f2) -- node[draw=none,auto=false][label=70: \Large $q$]{\Large\ldots} (f6);
										
					\node[fill=black, circle] (g1) at (0, -3) {};
					\node[fill=black, circle] (g2) at (0.5, -3) {};
					\node[fill=black,  ] (g6) at (3, -3) {};
					\path (g2) -- node[draw=none,auto=false][label=70: \Large $q$]{\Large\ldots} (g6);
					
					\node[fill=black, circle] (h1) at (0, -6) {};
					\node[fill=black, circle] (h2) at (0.5, -6) {};
					\node[fill=black, circle] (h6) at (3, -6) {};
					\path (h2) -- node[draw=none,auto=false][label=70: \Large $q$]{\Large\ldots} (h6);
    \node[fsnode,label=below left: \Large $z_1$](n1) at (1.5, -9)  {};
					\node[fill=black, circle] (s1) at (9.5, 0) {};
					\node[fill=black, circle] (s2) at (10.5, 0) {};
					\node[fill=black, circle] (s6) at (12.5, 0) {};
					\path (s2) -- node[draw=none,auto=false][label=70: \Large $q$]{\Large\ldots} (s6);
					\node[fsnode,xshift=11cm,yshift=-3cm,label=right: \Large $z_2$](n2) {};
					\node[fsnode,xshift=11cm,yshift=-5cm,label=right: \Large $z_3$](n3) {};
					
					\node[fill=black, circle] (t1) at (9.5, -9) {};
					\node[fill=black, circle] (t2) at (10.5, -9) {};
					\node[fill=black, circle] (t6) at (12.5, -9) {};
					\path (t2) -- node[draw=none,auto=false][label=70: \Large $q$]{\Large\ldots} (t6);
					
					\node [label=130: \Large\textbf{$Q_1$}] [myblue,fit=(f1) (f6), line width=2pt](c1) {};
					\node  [label=130: \Large\textbf{$Q_2$}] [myblue,fit=(s1) (s6), line width=2pt](c2) {};
					
					\node  [label=130: \Large\textbf{$Q_5$}][myblue,fit=(g1) (g6), line width=2pt](c3) {};
					\node  [label=130: \Large\textbf{$Q_4$}] [myblue,fit=(h1) (h6), line width=2pt](c4) {};
					\node  [label=110: \Large\textbf{$Q_3$}] [myblue,fit=(t1) (t6), line width=2pt](c5){};
					\path[every node/.style={font=\sffamily\small}]
					(n1) edge node {} (c2)
					edge node {} (n2)
					edge  node {} (n3)
					edge  node {} (c5);
					\draw (c1) -- (c2);
					\draw (c1) -- (n2);
					\draw (c3) -- (n2);
					\draw (c3) -- (n3);
					\draw (c4) -- (c5);
					\draw  (c4) -- (n3);
				\end{tikzpicture}}}} &

  \parbox{.4\textwidth}{\centering
					
		\scalebox{.6}{
		 	\resizebox{.8\textwidth}{!}{
				\begin{tikzpicture}[thick, auto,
					every node/.style={draw,circle,inner sep=3pt},
					fsnode/.style={fill=black},
					ssnode/.style={fill=black},
					every fit/.style={ellipse,draw,inner sep=16pt,text width=2.2cm},shorten >= 1pt,shorten <= 1pt
					]
					
					\node[fill=black, circle] (f1) at (0, 0) {};
					\node[fill=black, circle] (f2) at (0.5, 0) {};
					\node[fill=black, circle] (f6) at (3, 0) {};
					\path (f2) -- node[draw=none,auto=false][label=70: \Large $q$]{\Large\ldots} (f6);
										
					\node[fill=black, circle] (g1) at (0, -3) {};
					\node[fill=black, circle] (g2) at (0.5, -3) {};
					\node[fill=black,  ] (g6) at (3, -3) {};
					\path (g2) -- node[draw=none,auto=false][label=70: \Large $q$]{\Large\ldots} (g6);
					
					\node[fill=black, circle] (h1) at (0, -6) {};
					\node[fill=black, circle] (h2) at (0.5, -6) {};
					\node[fill=black, circle] (h6) at (3, -6) {};
					\path (h2) -- node[draw=none,auto=false][label=70: \Large $q$]{\Large\ldots} (h6);
    \node[fsnode,label=below left: \Large $z_1$](n1) at (1.5, -9)  {};
					\node[fill=black, circle] (s1) at (9.5, 0) {};
					\node[fill=black, circle] (s2) at (10.5, 0) {};
					\node[fill=black, circle] (s6) at (12.5, 0) {};
					\path (s2) -- node[draw=none,auto=false][label=70: \Large $q$]{\Large\ldots} (s6);
					\node[fsnode,xshift=11cm,yshift=-3cm,label=right: \Large $z_2$](n2) {};
					\node[fsnode,xshift=11cm,yshift=-5cm,label=right: \Large $z_3$](n3) {};
					
					\node[fill=black, circle] (t1) at (9.5, -9) {};
					\node[fill=black, circle] (t2) at (10.5, -9) {};
					\node[fill=black, circle] (t6) at (12.5, -9) {};
					\path (t2) -- node[draw=none,auto=false][label=70: \Large $q$]{\Large\ldots} (t6);
					
					\node [label=130: \Large\textbf{$Q_1$}] [myblue,fit=(f1) (f6), line width=2pt](c1) {};
					\node  [label=130: \Large\textbf{$Q_2$}] [myblue,fit=(s1) (s6), line width=2pt](c2) {};
					
					\node  [label=130: \Large\textbf{$Q_5$}][myblue,fit=(g1) (g6), line width=2pt](c3) {};
					\node  [label=130: \Large\textbf{$Q_4$}] [myblue,fit=(h1) (h6), line width=2pt](c4) {};
					\node  [label=110: \Large\textbf{$Q_3$}] [myblue,fit=(t1) (t6), line width=2pt](c5){};
					\path[every node/.style={font=\sffamily\small}]
					(n1)edge[bend left=90] node {} (c1)
					edge node {} (c2)
					edge node {} (n2)
					edge  node {} (n3)
					edge [bend left=70] node {} (c3)
					edge node {} (c4)
					edge  node {} (c5)
					(c2)edge node {} (n2)
					edge [bend left=50,red, line width=4pt] node {} (c5)
					edge [bend left] node {} (n3)
					(n2)edge node {} (n3)
					edge [bend left] node {} (c5);
					\draw (n3) -- (c5);
					\draw[red, line width=4pt] (c1) -- (c3);
					\draw[red, line width=4pt] (c3) -- (c4);
					\draw [red, line width=4pt](c1) -- (c2);
					\draw (c1) -- (n2);
					\draw (c3) -- (n2);
					\draw (c3) -- (n3);
					\draw [red, line width=4pt] (c4) -- (c5);
					\draw  (c4) -- (n3);
				\end{tikzpicture}}}} \\
   
    \vspace{.1in}
				\parbox{.4\textwidth}{\centering $H$} &  \vspace{.1in} \parbox{.4\textwidth}{\centering $H^2$}
			\end{tabular}
		\caption{A convex bipartite graph $H$ with $\omega(H^2)=2q + 3$ and $\chi(H^2)=\frac{5}{2}q+2$} \label{fig:Example}
	\end{figure}
	
	\medskip
	
	\section{A $\frac{3}{2}$-approximation algorithm}\label{sec:algo}
	Here, we propose a polynomial-time algorithm to find a proper coloring for squares of convex bipartite graphs using at most $3\omega/2$ colors (which is the same as the $\chi$-binding function proved in Theorem~\ref{thm:theorem-convex-bipartite}).
	
	\smallskip
	
	Let $G=(A,B,E)$ be a convex bipartite graph. 
 Recall the definitions of orderings $<_A$ and $<_B$. Note that the convexity property of $G$ is the well-known \emph{``consecutive-ones"} property when $G$ is represented as an adjacency matrix. Hence, the ordering $<_B$ of the vertices of $B$ can be obtained by using an $O(|V|+|E|)$-time algorithm proposed by Booth and Lueker~\cite{booth1976}. For ordering the vertices of $A$ with respect to $<_A$, we can use the \emph{``compact representation"} of $G$~\cite{klemz2022}, which can be computed in $O(|V|+|E|)$ time. In the compact representation of $G$, for each vertex $i \in A$, we have a triple $(i, \leftof(i),\rightof(i))$, where $\leftof(i)$ and $\rightof(i)$ are the least indexed and highest indexed neighbors of
	$i$ in $B$. Consequently, the collection of intervals $\{[\leftof(i),\rightof(i)]:i\in A\}$ forms a valid interval representation of the interval graph $G^2[A]$ (by Observation~\ref{obs:interval}). For $j\in \{1,2,\ldots,n\}$, recall the definitions of subgraphs $H_j$ of $G^2$. Our algorithm consists of two phases. First, we find a proper coloring of the interval graph $G^2[A]$ using the well-known greedy coloring algorithm~\cite{golumbic}. In the second phase, for each $j$ down from $n$ to $1$, we iteratively find a proper coloring for the subgraphs $H_j$ of $G^2$ using at most $3\omega/2$ colors. Since $H_1=G^2$, we then have the desired coloring of $G^2$. 


\subsection*{The algorithm}

 \noindent\textbf{Phase~I.} Find a proper coloring, say $c$ of $G^{2}[A]$ using the greedy algorithm.

  \medskip

  \noindent\textit{\textbf{Comment:}} Clearly, $c$ uses at most $\omega$ colors. Let $H_{n+1}=G^2[A]$ and $j\in \{1,2,\ldots,n\}$. For a non-extendable special coloring $c$ of $H_{j+1}$, recall the definitions of \emph{the pivot vertex, the pivot color, the partner color, the Kempe component containing pivot and partners, and the recoloring $\phi_c$} from the previous section. 
  Moreover, in the following phase of the algorithm, we say that $w$ is a \emph{free color} in  $S\subseteq V(G)$ with respect to a coloring $c$ of $H_{j+1}$ if $w\in \{1,2,\ldots,\lfloor\frac{3\omega}{2}\rfloor\}$ and $c(u)\neq w$ for any $u\in S$.  
\begin{tabbing}
    \noindent\textbf{Phase~II.}  \= For \=  each  \=  $j$ \= down from $n$ to $1$  do the \=following: 
    \\[.05in]\> 1. Consider the coloring $c$ of $H_{j+1}$.
    \\[.05in]\> 2. If there is a free color, say $w$ in $N_{H_j}(b_j)\subseteq V(H_{j+1})$ with respect to $c$, \\\>\> then assign $c(b_j)=w$.
    \\[0.05in]\> 3. Else, find the \emph{pivot vertex} $a_c$ with respect to $c$. Let $c(a_c)=x$.
    \\[.05in]\>\> 3.1. If there is a free color, say $z$ in $N_{H_{j+1}}[a_c]$, \\\>\>\> then assign  $c(a_c)=z$ and $c(b_j)=x$.
    \\[.05in]\>\> 3.2. Else, let $y$ be the \emph{partner color} of $a_c$. 
    \\[.05in]\>\>\hspace{0.6cm}(i) Find the lowest indexed vertex, say $a'$ in $N_{H_{j+1}}(a_c) \cap A$ such that \\\>\>\>\> \hspace{0.6cm}  $a' <_{A} a_c$ and $c(a')=y$.
    \\[.05in]\>\>\hspace{0.6cm}(ii) Swap the colors in the \emph{Kempe component (defined by the pivot and} \\\>\>\>\> \hspace{0.6cm} \emph{partner colors) containing the pivot}; i.e. set $c=\phi_c$.
    \\[.05in]\>\>\hspace{0.55cm} (iii) If $a'\notin N_{H_j}(b_j)$ then $c(b_j)=x$.
    \\[.05in]\>\>\hspace{0.55cm} (iv) Else, $a'$ is the new pivot vertex with respect to $c$ with $c(a')=x$;\\\>\>\>\> \hspace{0.6cm} set $a_c=a'$ and goto 3.1.
\end{tabbing}

		\medskip
  The above algorithm may not produce an optimal coloring, but the proof of Lemma~\ref{lem:main-convex-bipartite} guarantees that the algorithm terminates after finite steps and outputs a proper coloring of $G^2$ using at most $3\omega/2$ colors in polynomial time. To be specific, we note the following:

  \smallskip
  
  For each $j\in \{1,2,\ldots,n\}$, $c$ is a special coloring of $H_{j+1}$. Therefore, \emph{if step~2 is not true} then it implies that $c$ is a non-extendable special coloring of $H_{j+1}$ (see Definition~\ref{def:extend}). Thus, by Claim~\ref{claim:unique}, the pivot vertex $a_c$ exists, and it is easy to see that $a_c$ can be found in polynomial time. Now, \emph{if step~3.1 is not true} then it implies that Claim~\ref{claim:pivot}\ref{it:a} is true. Therefore, by Claim~\ref{claim:pivotcolor}, the partner color $y$ exists. Also, the steps 3.2.(i) and 3.2.(ii) can be executed in polynomial time. Note that for fixed $j$, the repetition of step 3.1 only happens when step~3.2.(iii) is not true (i.e. $a'\in N_{H_j}(b_j)$). Moreover, in this case, the new pivot vertex $a'$ obtained in step~3.2.(iv) has the property that $a'<a_c$, where $a_c$ is the current pivot vertex. Therefore, step 3 will be executed at most $|N_{H_j}(b_j)\cap A|$ times. 

  \medskip
  
  Let ${OPT}$ denote the number of colors used in an optimal coloring of $G^2$. Clearly, ${OPT} \geq \omega$. 
        Since our algorithm uses only at most $3\omega/2\leq \frac{3}{2} OPT$ colors, we have the following theorem.
		\begin{theorem}
			There exists a polynomial-time algorithm to find a proper coloring of squares of convex bipartite graphs with approximation ratio $\frac{3}{2}$.
		\end{theorem}

					

		\section{Relation with maximum degree}
		In this section, we prove the following theorem.
		
		\begin{theorem}\label{thm:degree}
			Let $G$ be a convex bipartite graph with maximum degree, $\Delta\geq 1$. Then $\chi(G^2)\leq 2\Delta$. Moreover, there exist convex bipartite graphs $G$ with $\chi(G^2)=2\Delta$.
		\end{theorem}
		Let $G=(A,B,E)$ be a complete bipartite graph with $|A|=|B|=n$. Clearly, $G$ is a convex bipartite graph with maximum degree $n$, and $G^2$ is a clique on $2n$ vertices. Thus we have $\chi(G^2)=2n$. This shows that the class of regular complete bipartite graphs form an instance of the graphs for which $\chi(G^2)=2\Delta$, where $\Delta$ is the maximum degree of $G$. This proves the latter part of Theorem~\ref{thm:degree}. In the remaining part of this section, we work towards proving the former part of Theorem~\ref{thm:degree}.
		
		\medskip
		
		First, we note the following observation.
		
		\begin{observation} \label{obs:ABclique}
			Let $G$ be a convex bipartite graph with orderings $<_A$ and $<_B$ (as defined earlier). Let $C_A$ be a clique in $G^2[A]$ and let $C_B$ be a clique in $G^2[B]$. Then there exists vertices $a\in A$ and $b\in B$, such that $C_A\subseteq N_G(b)$ and $C_B\subseteq N_G(a)$.     
		\end{observation}
		\begin{proof}
			Let $C_A$ be a clique in $G^2[A]$ with $|C_A|=p$ for some integer $p\geq 1$. Then, we can denote
			$C_A=\{a_{i_1},a_{i_2},\ldots,a_{i_p}\}$, where $a_{i_1}<_A a_{i_2}<_A\cdots<_A a_{i_p}$. Let $b\in B$ be the neighbor of $a_{i_1}$ in $G$ that has a maximum index with respect to the ordering $<_B$. i.e. $b=\max_{<_B}\{b_l: b_l\in N_G(a_{i_1})\}$. Let $a\in C_A\setminus \{a_{i_1}\}$. Since $C_A$ is a clique in $G^2[A]$, we have $a_{i_1}a\in E(G^2)$. This implies that $N_{G}(a_{i_1})\cap N_G(a)\neq \emptyset$. Since $a_{i_1}<_A a$, by the definitions of $<_A$ and $b$, we then have $b\in N_G(a)$. This implies that $C_A\subseteq N_G(b)$. 
			
			\smallskip
			
			Let $C_B$ be a clique in $G^2[B]$ with $|C_B|=q$ for some integer $q\geq 1$. Then we can denote $C_B=\{b_{j_1},b_{j_2},\ldots,b_{j_q}\}$, where $b_{j_1}<_B b_{j_2}<_B\cdots<_B b_{j_q}$. Since $b_{j_1}b_{j_q}\in E(G^2)$, there exists a vertex $a\in A$ such that $b_{j_1},b_{j_q}\in N_G(a)$. By the consecutive property of the ordering $<_B$, we then have $C_B=\{b_{j_1},b_{j_2},\ldots,b_{j_q}\}\subseteq N_G(a)$. 
		\end{proof}
		
		We are now ready to prove Theorem~\ref{thm:degree}.
		
		\medskip
		
		\noindent{\textbf{Proof of Theorem~\ref{thm:degree}} }
		
		\begin{proof}
			Let $\omega_A$ and $\omega_B$ denote the size of maximum cliques in $G^2[A]$ and $G^2[B]$ respectively. Then by Observation~\ref{obs:ABclique}, we have that $\omega_A\leq \Delta$ and $\omega_B\leq \Delta$. Recall that $G^2[A]$ is an interval graph by Observation~\ref{obs:interval}, and $G^2[B]$ is a proper interval graph by Lemma~\ref{obs:properinterval}. As both the subgraphs $G^2[A]$ and $G^2[B]$ are perfect, we then have $\chi(G^2[A])=\omega_A\leq \Delta$ and $\chi(G^2[B])=\omega_B\leq \Delta$. Therefore, we can conclude that $\chi(G^2)\leq \chi(G^2[A])+\chi(G^2[B])\leq 2\Delta$. This completes the proof of Theorem~\ref{thm:degree}.
			
		\end{proof}

		\section{On partite testable properties}
		Here, we present a few structural observations on the squares of bipartite graphs when some special properties are satisfied. In particular, we investigate some \emph{partite testable properties} (see Definition~\ref{def:partitetest}) for the squares of bipartite graphs.  
		
		\begin{theorem}\label{thm:oddantihole}
			The property of not containing odd anti-holes of size larger than five is a partite testable property for the squares of bipartite graphs.  
   \end{theorem}
		\begin{proof}
			Let $G=(A,B,E)$ be a bipartite graph with partite independent sets $A$ and $B$. Suppose that both the subgraphs $G^2[A]$  and $G^2[B]$ do not contain odd anti-holes of size larger than five. We need to prove that $G^2$ also satisfies the same property.  For the sake of contradiction, suppose that $G^2$ has an odd anti-hole, say $H$, of size larger than five. Let $V(H)=\{v_0,v_1,\ldots,v_k\}$, where $k$ is even, $k\geq 6$; here, the vertices are labeled with respect to the cyclic order of the vertices in its complement, which is a hole. (Throughout the proof, we consider the indices modulo $k+1$). Note that by the definition of an anti-hole, for each $i\in \{0,1,2,\ldots,k\}$, we have $N_H(v_i)=V(H)\setminus \{v_{i-1},v_{i+1}\}$. In other words, any two vertices that are non-adjacent in $H$ are consecutive with respect to their labeling.
			
			\smallskip
			We now have the following claims.
			
			\smallskip
			\begin{myclaim}\label{claim:nonadj}
				\textit{There exists $i\in \{0,1,\ldots,k\}$ such that $v_i,v_{i+1}\in A$ or $v_i,v_{i+1}\in B$}
			\end{myclaim}  
			\begin{proof}
                Suppose not. Without loss of generality, we can assume that $v_0\in A$. For $i\in \{0,1,\ldots,k\}$, we then have $v_i\in A$, if $i$ is even and $v_i\in B$, if $i$ is odd. As $k$ is even, this implies that $v_0,v_k\in A$. Since $v_{k+1}=v_0$, this is a contradiction to our assumption. Hence, the claim.
			\end{proof}
			\begin{myclaim}\label{claim:commonneighbors}
				\textit{For $i\in \{0,1,\ldots,k\}$, if $v_i,v_{i+1}\in A$ (respectively, $B$) then $N_H(v_i)\cap N_H(v_{i+1}) \subseteq A$ (respectively, $B$).} 
			\end{myclaim}
			\begin{proof}
				It is enough to prove the claim for the case in which $v_i,v_{i+1}\in A$ (as the proof for the respective case is similar). If possible, assume that there exists a vertex $b\in N_H(v_i)\cap N_H(v_{i+1})$ such that $b\in B$. 
				This implies that $v_ib,v_{i+1}b\in E(G)$ (since $v_i,v_{i+1}\in A$ and by Observation~\ref{obs:bipsquare}). But we then have  $v_iv_{i+1}\in E(G^2)$,  a contradiction. Therefore, we can conclude that $N_H(v_i)\cap N_H(v_{i+1})\subseteq A$. Hence, the claim.
			\end{proof}
			\medskip
			Without loss of generality, by Claim~\ref{claim:nonadj}, we can assume that $v_i,v_{i+1} \in A$, for some $i\in \{0,1,\ldots,k\}$. Then by Claim~\ref{claim:commonneighbors}, we have $V(H)\setminus \{v_{i-1},v_{i+2}\}=(N_H(v_i)\cap N_H(v_{i+1})) \cup \{v_i,v_{i+1}\} \subseteq A$. Now, we will show that even the remaining vertices $v_{i-1}$ and $v_{i+2}$ are also in $A$. Since $k\geq 6$, we can find distinct vertices $v_{i+3},v_{i+4},v_{i+5}\in  V(H)\setminus \{v_i,v_{i+1},v_{i-1},v_{i+2}\}$. Since $V(H)\setminus \{v_{i-1},v_{i+2}\} \subseteq A$, we then have $v_{i+3},v_{i+4},v_{i+5} \in A$. Note that $v_{i-1}\in N_H(v_{i+3})\cap N_H(v_{i+4})$ and $v_{i+3},v_{i+4} \in A$. Similarly, $v_{i+2} \in N_H(v_{i+4})\cap N_H(v_{i+5})$ and $v_{i+4},v_{i+5} \in A$. Therefore, again by Claim~\ref{claim:commonneighbors}, we have $v_{i-1},v_{i+2}\in A$. This together with the fact that $V(H)\setminus \{v_{i-1},v_{i+2}\} \subseteq A$ implies that $V(H)\subseteq A$. Recall that $H$ is an anti-hole of size larger than five. This implies that $G^2[A]$ contains an odd anti-hole of size larger than five as an induced subgraph. As this is a contradiction, we can conclude that $G^2$ does not contain odd anti-holes of length larger than five. Hence, the theorem.
		\end{proof}

		Let $G=(A,B,E)$ be a chordal bipartite graph. By a result in~\cite{Le2019hardness}, both the subgraphs, $G^2[A]$ and $G^2[B]$ (which they refer to as \emph{half squares} in~\cite{Le2019hardness}) are chordal and therefore, perfect. Since perfect graphs are (odd antihole)-free, we then have the following corollary due to Theorem~\ref{thm:oddantihole}.
		\begin{corollary}\label{corr:chordalbip}
			Squares of chordal bipartite graphs do not contain odd anti-holes of size larger than five.
		\end{corollary}
		In the remarks below, we note that the properties, namely, being \emph{(even anti-hole)-free or (hole)-free} are not partite testable properties for the squares of bipartite graphs in general. 
		\begin{remark}\textbf{\textit{(Even anti-hole)-free property is not partite testable in general:}}
			Let $G=(A,B,E)$ be a bipartite graph with partite independent sets $A$ and $B$. Even if both the induced subgraphs $G^2[A]$  and $G^2[B]$ do not contain even anti-holes of size larger than four, $G^2$ may still contain even anti-holes of any size larger than four. See Figure~\ref{antihole} for an illustration (the edges of an even anti-hole in $G^2$ are shown in red). To be precise, here, the induced subgraphs $G^2[A]$ and $G^2[B]$ do not contain any anti-holes (as both the subgraphs induce cliques in $G^2$). But in $G^2$, the set of vertices $\{a_2,b_2,a_3,b_3,a_4,b_4\}$ induces an even anti-hole of size six. 
		\end{remark}
		
		\begin{figure}[h]
			\begin{tabular}{p{.4\textwidth}p{.4\textwidth}}
				\parbox{.4\textwidth}{\centering
					\begin{tikzpicture}[scale=.75]
						\node[draw,circle] (a1) at (0,0) {$a_1$};
						\node[draw,circle] (a2) at (0,-1.5) {$a_2$};
						\node[draw,circle] (a3) at (0,-3) {$a_3$};
						\node[draw,circle] (a4) at (0,-4.5) {$a_4$};
						\node[draw,circle] (b1) at (2,0) {$b_1$};
						\node[draw,circle] (b2) at (2,-1.5) {$b_2$};
						\node[draw,circle] (b3) at (2,-3) {$b_3$};
						\node[draw,circle] (b4) at (2,-4.5) {$b_4$};
						\draw (a1) -- (b2);
						\draw (a1) -- (b3);
						\draw (a1) -- (b4);
						\draw (a2) -- (b1);
						\draw (a2) -- (b3);
						\draw (a3) -- (b1);
						\draw (a3) -- (b4);
						\draw (a4) -- (b1);
						\draw (a4) -- (b2);
				\end{tikzpicture}} &
				\parbox{.4\textwidth}{\centering
					\begin{tikzpicture}[scale=.75]
						\node[draw,circle] (a1) at (0,0) {$a_1$};
						\node[draw,circle] (a2) at (0,-1.5) {$a_2$};
						\node[draw,circle] (a3) at (0,-3) {$a_3$};
						\node[draw,circle] (a4) at (0,-4.5) {$a_4$};
						\node[draw,circle] (b1) at (2,0) {$b_1$};
						\node[draw,circle] (b2) at (2,-1.5) {$b_2$};
						\node[draw,circle] (b3) at (2,-3) {$b_3$};
						\node[draw,circle] (b4) at (2,-4.5) {$b_4$};
						\draw (a1) -- (b2);
						\draw (a1) -- (b3);
						\draw (a1) -- (b4);
						\draw (a2) -- (b1);
						\draw[color=red,line width=1.5pt] (a2) -- (b3);
						\draw (a3) -- (b1);
						\draw[color=red,line width=1.5pt] (a3) -- (b4);
						\draw (a4) -- (b1);
						\draw[color=red,line width=1.5pt] (a4) -- (b2);
						\draw (a1) -- (a2);
						\draw[color=red,line width=1.5pt](a2) -- (a3);
						\draw[color=red,line width=1.5pt](a3) -- (a4);
						\draw(a1).. controls (-.75,-2).. (a3);
						\draw(a1).. controls (-1.25,-2.75).. (a4);
						\draw[color=red,line width=1.5pt] (a2).. controls (-.75,-3).. (a4);
						\draw (b1) -- (b2);
						\draw[color=red,line width=1.5pt] (b2) -- (b3);
						\draw [color=red,line width=1.5pt](b3) -- (b4);
						\draw (b1).. controls (2.75,-2).. (b3);
						\draw(b1).. controls (3.25,-2.75).. (b4);
						\draw[color=red,line width=1.5pt] (b2).. controls (2.75,-3).. (b4);
				\end{tikzpicture}}\\
				\vspace{.1in}
				\parbox{.4\textwidth}{\centering $G$} &  \vspace{.1in} \parbox{.4\textwidth}{\centering $G^2$}
			\end{tabular}
			\caption{An illustration to show that \emph{(even-hole)-free} is not a partite testable property}
			\label{antihole}
		\end{figure}
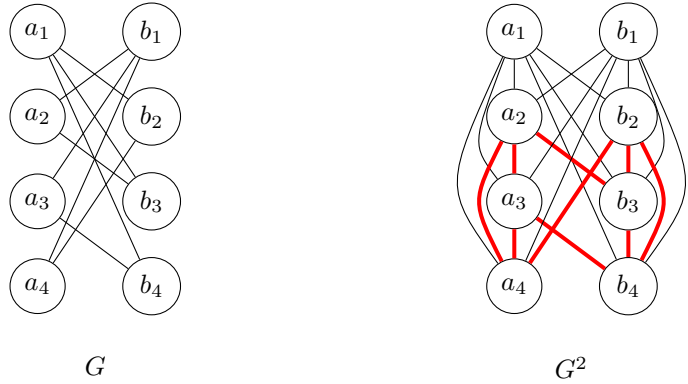
		
		\begin{remark}\textbf{\textit{(Even hole)-free or (odd hole)-free properties are not partite testable in general:}}
			Let $G=(A,B,E)$ be a convex bipartite graph with partite independent sets $A$ and $B$. Even if both the subgraphs $G^2[A]$ and $G^2[B]$ do not contain holes of length greater than three, $G^2$ may contain holes of any length greater than three (see Figure~\ref{fig:not_perfect} for an example). 
		\end{remark}
		\medskip
		
		Note that the structure of the graph $G^2$ given in Figure~\ref{fig:not_perfect} is particularly interesting. Later, we will see in Theorem~\ref{cycle:structure} that any induced cycle of length at least four that is present in the square of a convex bipartite graph follows the same structure as the one in Figure~\ref{fig:not_perfect}; see the details in Theorem~\ref{cycle:structure} for a precise statement. 

        \medskip
        
        Given a convex bipartite graph $G=(A,B,E)$, recall the definitions of orderings $<_A$ and $<_B$ for the vertices in $A$ and $B$, respectively. In the rest of the section, for $b_i,b_j\in B$ with $b_i<_B b_j$, we denote by $[b_i,b_j]$, the set $\{b\in B: b_i\leq_B b \leq_B b_j\}$. Further, we denote by $(b_i,b_j)$ the set $[b_i,b_j]\setminus \{b_i,b_j\}$. Let $G=(A,B,E)$ be a bipartite graph. For $k\geq 2$, we call an induced path $P=(a_1,a_2,\ldots,a_k)$ in $G^2$ an \emph{$(A,b,b')$-path}, if $V(P)\subseteq A$ and there exist distinct vertices $b,b'\in B$ such that $b\in N_G(a_1)\setminus \Big(\bigcup\limits_{j=2}^k N_G(a_j)\Big)$ and $b'\in N_G(a_k)\setminus \Big(\bigcup\limits_{j=1}^{k-1} N_G(a_j)\Big)$. Note that if $P$ is an $(A,b,b')$-path, then $V(P)\cup\{b,b'\}$ may induce either a path or cycle in $G^2$. First, we note some properties of \emph{$(A,b,b')$-paths} that are present in the square of a convex bipartite graph.
	\begin{observation}\label{obs:path}
		Let $G=(A,B,E)$ be a convex bipartite graph with the orderings $<_A$ and $<_B$. For $k\geq 2$, let $P=(a_1,a_2,\ldots,a_k)$ be an $(A,b,b')$-path in $G^2$ with $a_1<_A a_k$. Then $b<_B b'$ and $a_1=\min_{<_A}\{a_j\in V(P)\}$.
	\end{observation}
\begin{proof}
 It is easy to see that $b<_B b'$, as otherwise, by Observation~\ref{obs:orderprop} (with $a=a_1$ and $a'=a_k$) we have
 $b\in N_G(a_k)$, and this contradicts the definition of $(A,b,b')$-path.
Now to prove the latter part, let $a_i=\min_{<_A}\{a_j\in V(P)\}$. Note that $a_1<_A a_k$. Suppose that $a_i=a_j$, for some $j\in \{2,3,\ldots,k-1\}$. Then the vertices $a_{j-1},a_{j+1}$ exist, and are such that $a_j<_A a_{j-1},a_{j+1}$, and $a_{j-1},a_{j+1} \in N_P(a_j)\subseteq N_{G^2[A]}(a_j)$. By Observation~\ref{obs:clique}, we then have $a_{j-1}a_{j+1}\in E(G^2[A])$. This contradicts the fact that $P$ is an induced path. Therefore, we can conclude that $a_i=a_1$.

\end{proof}
		
		\begin{lemma}\label{lem:path}
			Let $G=(A,B,E)$ be a convex bipartite graph with the orderings $<_A$ and $<_B$. For $k\geq 2$, let $P=(a_1,a_2,\ldots,a_k)$ be an $(A,b,b')$-path in $G^2$ with $a_1<_A a_k$. Then the following hold:
			\begin{myenumerate}
				\item \label{item:neighbor} If $k=2$, we have $N_G(a_1)\cap N_G(a_k)\subseteq (b,b')$. If $k>2$, for each $j\in \{2,3,\ldots,k-1\}$,  we have $N_G(a_j)\subseteq (b,b')$. 
				\item \label{item:span} For each $\Tilde{b} \in [b,b']$, there exists $j\in \{1,2,\ldots,k\}$ such that $\Tilde{b}\in N_G(a_j)$.
			\end{myenumerate}
		\end{lemma}
		\begin{proof}
		  

			
			
			\noindent \ref{item:neighbor} Note that $b<_B b'$ by Observation~\ref{obs:path}. If $k=2$, then as $a_1a_k\in E(G^2)$, $N_G(a_1)\cap N_G(a_k)\neq \emptyset$. Further, by the consecutive property of the ordering $<_B$, and the definitions of $b$ and $b'$, we have $N_G(a_1)\cap N_G(a_k) \subseteq (b,b')$. Suppose that $k>2$. If possible, let $a_j$, where $j\in \{2,3,\ldots,k-1\}$ be the first vertex along the path $P$ for which \ref{item:neighbor} is not true. i.e. let $j= \min \{2,3,\ldots,k-1\}$ be such that $N_G(a_j)\nsubseteq (b,b')$.

   \smallskip
			
			
		 Recall that $b,b'\notin N_G(a_j)$. Thus by the definition of $<_B$, if $a_j$ has a neighbor $b_l\in (b,b')$ then we have $N_G(a_j)\subseteq (b,b')$, a contradiction, and we are done. Therefore we can assume that  $N_G(a_j)\cap [b,b']=\emptyset$. But then it should be the case that $j=2$. (As otherwise, $j-1\geq 2$ and by the choice of $j$ we have  $N_G(a_{j-1})\subseteq (b,b')$. This implies that $N_G(a_{j-1})\cap N_G(a_j)=\emptyset$, and contradicts the fact that $a_{j-1}a_j\in E(G^2)$.) Now since $b'\notin N_G(a_1)$ and $a_1a_{2}\in E(G^2)$, $N_G(a_2)\nsubseteq [b,b']$ implies that $N_G(a_2)\subseteq [b_1,b-1]$, where  $b_1=\min_{<_B}\{b\in B\}$. Since $a_2a_{3}\in E(G^2)$ and $b\notin N_G(a_{3})$, by the consecutive property of $<_B$ we then have $N_G(a_{3})\subseteq [b_1,b-1]$. For each $l\geq 3$, by repeating the same arguments along the edges $a_la_{l+1}$ in $G^2$,  we finally get $N_G(a_k)\subseteq [b_1,b-1]$. As  $b<_B b'$, the consecutive property of $<_B$ implies that $b'\notin N_G(a_k)$, a contradiction. Therefore, \ref{item:neighbor} is true.
			
			\medskip
			
			\noindent\ref{item:span} Suppose not. Let $x \in [b, b']$ be such that for any $j\in \{1,2,\ldots,k\}$, $x\notin N_G(a_j)$. Consider the vertex $a_2$. By \ref{item:neighbor}, $N_G(a_2)\subseteq (b,b')$. Since $a_1a_2\in E(G^2)$ there exists a vertex $y\in B$ such that $a_1y,a_2y\in E(G)$. As $b\in N_G(a_1)$, $y\in N_G(a_1)\cap N_G(a_2)$ and $x\notin N_G(a_1)\cup N_G(a_2)$, by the consecutive property of the ordering $<_B$, we have $N_G(a_2)\subseteq (b,x)$. Now for each $l\geq 2$, by repeating the same arguments along the edges $a_la_{l+1}$ in $G^2$, we finally get $N_G(a_k)\subseteq (b,x)$. Since $x <_B b'$, the consecutive property of $<_B$ implies that $b'\notin N_G(a_k)$. This is a contradiction and hence proves \ref{item:span}.
		\end{proof}
		We now have the following corollary of Lemma~\ref{lem:path}.
		
		\begin{corollary}\label{corr:path}
		Let $G=(A,B,E)$ be a convex bipartite graph. Let $C$ be an induced cycle in $G^2$ and let $P=(a_1,a_2,\ldots,a_k)$ (where $k\geq 2$), be a sub-path (recall the definition from Section~\ref{sec:prelim}) of $C$ which is also an $(A,b,b')$-path. If  $b,b'\in B\cap V(C)$ then $(b, b')\cap V(C)=\emptyset$.
		\end{corollary}
  \begin{proof}
      Suppose not. Let  $x\in (b, b')\cap V(C)$. Then by Lemma~\ref{lem:path}\ref{item:span}, there exists a vertex $a_j\in V(P)$ such that $xa_j\in E(G)\subseteq E(G^2)$. This contradicts the fact that $C$ is an induced cycle.
  \end{proof}
		
		\medskip
		
		Let $G$ be a convex bipartite graph. Since the class of convex bipartite graphs forms a subclass of chordal bipartite graphs, by Corollary~\ref{corr:chordalbip}, $G^2$ does not contain odd antiholes of size larger than five. Therefore, by the \textit{Strong Perfect Graph Theorem}, the sole reason for the non-perfectness of the squares of convex bipartite graphs is due to the presence of odd holes (note that an odd antihole of size five is isomorphic to its complement $C_5$, the odd hole of size five). In the following observations and lemmas, we study the structure of holes that are present in the squares of convex bipartite graphs.
		
		\smallskip
		\begin{observation}\label{obs:cycleneighbors}
		Let $G=(A,B,E)$ be a convex bipartite graph with the orderings $<_A$ and $<_B$, and let $C$ be an induced cycle in $G^2$ of length $k\geq 4$. Then $V(C)\cap A\neq \emptyset$ and $V(C)\cap B\neq \emptyset$. Moreover, if the vertices in $C$ are labeled such that $C=(v_1,v_2,\ldots,v_k)$ with $v_1=\min_{<_A}\{v_l:v_l\in V(C)\cap A\}$, then exactly one of the vertices in the set $\{v_{2},v_{k}\}$ belongs to $A$ and the other is in $B$.
		\end{observation}
	\begin{proof}
  By Observation~\ref{obs:interval} and Lemma~\ref{obs:properinterval}, both the subgraphs $G^2[A]$ and $G^2[B]$ are interval graphs and therefore, chordal. Since $C$ is an induced cycle of length at least 4 in $G^2$, we then have $V(C)\nsubseteq A$ and $V(C)\nsubseteq B$. This implies that $V(C)\cap A \neq \emptyset$ and $V(C)\cap B \neq \emptyset$. Let $v_1=\min_{<_A}\{v_l:v_l\in V(C)\cap A\}$. Now consider the vertices $v_{2}$ and $v_{k}$ in $C$. If $\{v_{2}, v_{k}\}\subseteq A$ then $v_1<_A v_{2}, v_{k}$ and $\{v_{2}, v_{k}\}\subseteq N_{C}(v_1)\subseteq N_{G^2}(v_1)$. This implies by Observation~\ref{obs:clique} that $v_{2}v_{k}\in E(G^2)$, a contradiction to the fact that $C$ is an induced cycle.  Now if $\{v_{2}, v_{k}\}\subseteq B$ then $\{v_{2},v_{k}\}\subseteq N_{C}(v_1)\subseteq N_G(v_1)$ (as $v_1\in A$ and $v_{2},v_{k}\in B$, by Observation~\ref{obs:bipsquare}). This implies that $v_{2}v_{k}\in E(G^2)$, again a contradiction. As we have a contradiction in both cases, we can conclude that exactly one of the vertices in $\{v_{2},v_{k}\}$ belongs to $A$ and the other is in $B$. 
		
	\end{proof}
	\begin{lemma}\label{lem:cycles}
			Let $G=(A,B,E)$ be a convex bipartite graph with the orderings $<_A$ and $<_B$ and let $C$ be an induced cycle in $G^2$ of length $k\geq 4$. Then the vertices in $C$ can be labeled as $C=(v_1,v_2,\ldots,v_k)$ with $v_1=\min_{<_A}\{v_l:v_l\in V(C)\cap A\}$ such that the sub-path $(v_1,v_2,\ldots,v_{k-2})$ is an $(A,v_k,v_{k-1})$-path. 
		\end{lemma}
		\begin{proof}
By Observation~\ref{obs:cycleneighbors}, if $v_1=\min_{<_A}\{v_l:v_l\in V(C)\cap A\}$ then exactly one of the vertices in $\{v_2,v_{k}\}$ belongs to $A$ and the other is in $B$. Without loss of generality, we can assume that $v_{2}\in A$ and $v_{k}\in B$ (as in the other case, we can first enumerate the cycle by starting at $v_1$ itself but in the reverse direction and use the same arguments below). Now, while traversing $C$ starting from $v_1\in A$ along the increasing order of their indices in $C$, namely, $v_1, v_{2},\ldots,$ in $A$, we can find a vertex, say $v_j\in V(C)\cap A$ such that $\{v_1,v_{2},\ldots, v_j\} \subseteq A$, but $v_{j+1}\in B$. If $v_{j+1}= v_{k}$ then $v_jv_{k},v_1v_{k}\in E(C)\subseteq E(G)$ (since $v_j,v_1\in A$ and $v_k\in B$,  by Observation~\ref{obs:bipsquare}). This implies that $v_1v_j\in E(G^2)$, a contradiction to the fact that $C$ is an induced cycle of length at least four in $G^2$. Thus we have $v_{j+1}\neq v_{k}$, and therefore $j\leq k-2$. Moreover, $v_{k}\in N_G(v_1)\setminus \Big(\bigcup\limits_{p=2}^{j} N_G(v_p)\Big)$ and $v_{j+1}\in N_G(v_j)\setminus \Big(\bigcup\limits_{p=1}^{j-1} N_G(v_p)\Big)$. Therefore,  $P_0= (v_1,v_{2},\ldots, v_j)$ is an  $(A,v_{k}, v_{j+1})$-path and by Observation~\ref{obs:path}, we have $v_{k}<_Bv_{j+1}$. Observe that if $v_{k}v_{j+1}\in  E(G^2)$ then $v_{j+1}=v_{k-1}$, $v_j=v_{k-2}$, and we are done. We intend to show that $v_{k}v_{j+1}\in  E(G^2)$ is always true.


     \medskip
     
     For the sake of contradiction, suppose that $v_{k}v_{j+1}\notin E(G^2)$. Then consider the sub-path of $C$ namely, $P_1: (v_{k},v_{k-1},\ldots,v_{j+2},v_{j+1})$ (starting from the vertex $v_{k}$ in $C$ traverse along the reverse direction all the way to $v_{j+1}$). Note that $|V(P_1)|\geq 3$ (since $v_{k}v_{j+1}\notin E(G^2)$) and $V(C)=V(P_0)\uplus V(P_1)$. Let $b_1$ be the minimum indexed vertex in $B$ with respect to the ordering $<_B$. We then claim that $V(P_1)\setminus \{v_{j+1}\}\subseteq [b_1,v_{k}]$. Suppose not. Let $v_p$ be the first vertex with respect to the order of their appearance in $P_1$ such that $v_p\notin [b_1,v_{k}]$. We then have the following cases.
     
\medskip
\noindent\textbf{Case-1 $v_p\in A$:} Recall that $v_1\in A$, $v_{k}\in B$. Therefore, $v_p\in A$ implies that $v_p\neq v_{k-1}$. As otherwise, $v_pv_k,v_1v_k\in E(C)\subseteq E(G)$ (by Observation~\ref{obs:bipsquare}). This implies that $v_pv_{1}\in E(G^2)$,  a contradiction (since $2<j+1<p<k$). Recall that $v_1<_A v_p$. By the choice of $v_p$ and the fact that $v_p\neq v_{k-1}$,  we also have $b_1\leq_B v_{p-1}<_B v_{k}$. Then as $v_1v_k,v_{p}v_{p-1}\in  E(C)\subseteq E(G)$ (as $v_1,v_p\in A$ and $v_{k},v_{p-1}\in B$, by Observation~\ref{obs:bipsquare}),  by Observation~\ref{obs:orderprop} (with $a=v_1$, $a'=v_p$, $b=v_{k}$, and $b'=v_{p-1}$), we get $v_pv_{k}\in E(G^2)$, a contradiction (since $v_p\neq v_{k-1},v_1$).

     \medskip
\noindent\textbf{Case-2 $v_p\in B:$} Here we have two subcases. 

\medskip

\noindent\textbf{Case-2.1 $v_p\in (v_{k},v_{j+1})$:}
 Since $P_0$ is an $(A,v_{k}, v_{j+1})$-path and $v_p\in V(C)$, by Corollary~\ref{corr:path} applied to $P_0$, we have $v_{p}\notin (v_{k},v_{j+1})$, a contradiction.

\medskip
\noindent\textbf{Case-2.2 $v_{j+1}<_B v_p$:} By the choice of $v_p$, we then have $b_1\leq_B v_{p-1}\leq_B v_{k}<_B v_{j+1}<_B v_p$. Now since $<_B$ is a proper interval ordering (by Lemma~\ref{obs:properinterval}), $v_{p-1}v_p\in E(G^2)$ implies that $v_{p-1}v_{j+1}, v_{j+1}v_p\in E(G^2)$, and therefore, $\{v_{p-1},v_p,v_{j+1}\}$ forms a triangle in $G^2$. This contradicts the fact that $C$ is an induced cycle of length at least 4 in $G^2$.


\medskip
Since we have a contradiction to the existence of $v_p$ in all the possible cases, we can conclude that $V(P_1)\setminus \{v_{j+1}\}\subseteq [b_1,v_{k}]$. Now consider the vertex $v_{j+2}\in V(P_1) \setminus \{v_{k},v_{j+1}\}$ (such a vertex $v_{j+2}$ exists since $|V(P_1)|\geq 3$). It follows that $v_{j+2}\in [b_1,v_{k}]$. This implies that $b_1<_Bv_{j+2}<_Bv_{k}<_Bv_{j+1}$. Further, as $<_B$ is a proper interval ordering (by Lemma~\ref{obs:properinterval}), the fact that $v_{j+2}v_{j+1}\in E(G^2)$ implies that $v_{k}v_{j+1}\in E(G^2)$, and therefore,  a contradiction. This completes the proof of the lemma.
\end{proof}
    		We are now ready to prove the following theorem on the structure of holes that is present in the squares of convex bipartite graphs.
		\begin{theorem}\label{cycle:structure}
			Let $G=(A,B,E)$ be a convex bipartite graph with the orderings $<_A$ and $<_B$. Let $C$ be an induced cycle in $G^2$ of length at least 4. Then the following hold:
			\begin{myenumerate}
				\item \label{P_1} The vertices in $C$ can be labelled as $(v_1,v_2\ldots,v_k)$, $k\geq 4$ such that the sub-path $P=(v_1,v_2,\ldots,v_{k-2})$ is an $(A,v_k,v_{k-1})$-path. 
    
				\item \label{P_2}  There exists a set of vertices $\{b_1,b_2,\ldots,b_{k-3}\} \subseteq (v_k,v_{k-1})$ such that for each $i\in \{1,2,\ldots,\\k-3\}$, we have $N_G(b_i)\cap V(P) = \{v_i,v_{i+1}\}$. Further, $v_k<_B b_1<_B\cdots<_B b_{k-3}<_B v_{k-1}$ and  $v_1<_A v_2<_A\cdots<_A v_{k-2}$.
    
				\item \label{P_3} There exists a vertex $a\in A$ such that $\{v_k,b_1,b_2,\ldots,b_{k-3},v_{k-1}\} \subseteq N_G(a)$. 
			\end{myenumerate}
		\end{theorem}
		\begin{proof}
			\ref{P_1} Let $v_1=\min_{<_A}\{v_i:v_i\in V(C)\cap A\}$. Then Part~\ref{P_1} is essentially same as Lemma~\ref{lem:cycles}.  
			
			\medskip
			
			\ref{P_2} 
 Let $b_0=v_k$ and $i\in \{1,2,\ldots,k-3\}$. Since $P$ is an induced path in $G^2[A]$, $N_G(v_i)\cap N_G(v_{i+1})\neq \emptyset$ for each $i$. Let $b_i\in N_G(v_i)\cap N_G(v_{i+1})$. Note that $b_i\notin N_G(v_j)$ for each $j\neq i,i+1$ (as otherwise $\{v_i,v_{i+1},v_j\}$ forms a triangle in $G^2$, a contradiction).
 
 We now recursively do the following procedure for $i\in \{1,2,\ldots,k-3\}$: Consider the sub-path $P_i=(v_i,v_{i+1},\ldots,v_{k-2})$. Clearly, $P_i$ is an $(A,b_{i-1},v_{k-1})$-path. And by Lemma~\ref{lem:path}\ref{item:neighbor}, $b_i\in (N_G(v_i)\cap N_G(v_{i+1}))\subseteq (b_{i-1},v_{k-1})$.
Therefore, we have the set of vertices $\{b_1,b_2,\ldots,b_{k-3}\}$ in $B$ having the property that $v_k<_B b_1<_B\cdots<_B b_{k-3}<_B v_{k-1}$ and $N_G(b_i)\cap V(P) = \{v_i,v_{i+1}\}$. Since for each $i\in \{1,2,\ldots,k-3\}$, $v_i\in N_G(b_i)\setminus N_G(b_{i+1})$, this implies by the definition of $<_A$ that $v_1<_A v_2<_A\cdots<_A v_{k-2}$. This proves~\ref{P_2}.

			\medskip
			
			\ref{P_3} Since $v_{k}v_{k-1}\in E(C)\subseteq E(G^2[B])$ (as $v_{k},v_{k-1}\in B$), there exists a vertex $a\in A$ such that $v_{k},v_{k-1}\in N_G(a)$. Recall that $v_{k}<_B b_1<_B\cdots<_B b_{k-3}<_B v_{k-1}$ (by~\ref{P_2}). This implies by the consecutive property of $<_B$ that $\{v_k,b_1,b_2,\ldots,b_{k-3},v_{k-1}\} \subseteq N_G(a)$, and proves~\ref{P_3}.
		\end{proof}
  We then have the following corollaries of Theorem~\ref{cycle:structure}.

 \begin{corollary}\label{corr:cyclepart}
    Let $G=(A,B,E)$ be a convex bipartite graph. Let $\{\overline{S}, S\}=\{A,B\}$, where $S$ is the partite set of $G$ whose vertices have an ordering with a consecutive property and $\overline{S}=V(G)\setminus S$. If $C$ is an induced cycle of length $k\geq 4$ then $|V(C)\cap S|=2$ and $|V(C)\cap \overline{S}|=k-2$.
 \end{corollary}
  \begin{proof}
   The vertices in the partite set $S$ of $G$ always admit an ordering as in the definition of the ordering $<_B$, and the vertices in $\overline{S}$ always admit an ordering as in the definition of the ordering $<_A$ (see the definitions of $<_A$ and $<_B$ defined at the beginning of Section~\ref{sec:convex_main}). The proof of the corollary now easily follows from Theorem~\ref{cycle:structure}\ref{P_1} and the definition of an $(A,b,b')$ path. 
  \end{proof}
		\begin{corollary} \label{corr:cycleC_5}
		Let $G$ be a convex bipartite graph. If $G^2$ contains an induced cycle of length $k\geq 4$, then $G^2$ contains induced cycles of length $k'$ for each $4\leq k'\leq k$.
		\end{corollary}
		\begin{proof}
			Let $C$ be an induced cycle in $G^2$ of length $k$, where $k\geq 4$. Then by Theorem~\ref{cycle:structure}\ref{P_1}, the vertices in $C$ can be labelled as $C= (v_1,\ldots,v_k)$, where $P=(v_1,v_2,\ldots,v_{k-2})$ is an $(A,v_k,v_{k-1})$-path. Further by  Theorem~\ref{cycle:structure}\ref{P_2}, there exists a set of vertices, say $B'=\{b_1,b_2,\ldots,b_{k-3}\}\subseteq B$ such that for each $j\in \{1,2,\ldots,k-3\}$, $N_G(b_j)\cap V(P) = \{v_j,v_{j+1}\}$. Also, by Theorem~\ref{cycle:structure}\ref{P_3}, the set $B'\cup \{v_{k},v_{k-1}\}$ forms a clique in $G^2$. Let $b_{k-2}=v_{k-1}$. For each $i\in \{2,\ldots,k-2\}$, define $C_i=(v_1,v_2,\ldots,v_i,b_i,v_k)$. As $C$ is an induced cycle and $N_{C_i}(b_i)=\{v_i,v_k\}$, this implies that $C_i$ is an induced cycle in $G^2$ of length $i+2$. Hence the corollary.
		\end{proof}
		Recall that a graph $H$ is said to be perfect if and only if $H$ is both (odd hole)-free and (odd antihole)-free; an odd antihole of length five is isomorphic to its complement $C_5$, the odd-hole of length five. 
		A graph $H$ is chordal if it does not contain induced cycles of length at least four. For a convex bipartite graph $G=(A,B,E)$, we have that both the subgraphs $G^2[A]$ and $G^2[B]$ are interval graphs and, therefore, chordal and hence perfect. Moreover, since the class of convex bipartite graphs forms a subclass of chordal bipartite graphs, by Corollary~\ref{corr:chordalbip}, we have that $G^2$ does not contain odd-antiholes of length greater than five. Now, the following theorem is an easy consequence of Corollary~\ref{corr:cycleC_5}.
		\begin{theorem}\label{thm:c4free}
		Let $G$ be a convex bipartite graph. Then, we have the following.
   \begin{myenumerate}
       \item \label{item:c5free} If $G^2$ is $C_5$-free, then $G^2$ is perfect. i.e. perfectness is a partite testable property for $C_5$-free squares of convex bipartite graphs.
       \item \label{item:c4free} If $G^2$ is $C_4$-free, then $G^2$ is chordal. i.e. chordality is a partite testable property for $C_4$-free squares of convex bipartite graphs.
   \end{myenumerate}

		\end{theorem}
  Now, let us focus on a well-known subclass of convex bipartite graphs, namely biconvex bipartite graphs.
		\begin{definition}[Biconvex bipartite graph]\label{def:biconvex}
			A bipartite graph $G=(A,B,E)$ is said to be biconvex bipartite if the vertices in $A$ and $B$, respectively, have orderings, say $<_A'$ and $<_B$ such that, for each vertex $a\in A$ (respectively, $b\in B$), the vertices in the neighborhood of $a$ (respectively, $b$) in $G$ appear consecutively with respect to the ordering $<_B$ (respectively, $<_A'$).
		\end{definition}
		\begin{figure}[h]
			\begin{tabular}{p{.4\textwidth}p{.4\textwidth}}
				\parbox{.4\textwidth}{\centering
					\begin{tikzpicture}[scale=.75]
						\node[draw,circle] (a2) at (0,-1.5) {$v_1$};
						\node[draw,circle] (a3) at (0,-3) {$a$};
						\node[draw,circle] (a4) at (0,-4.5) {$v_2$};
						\node[draw,circle] (b2) at (2,-1.5) {$v_4$};
						\node[draw,circle] (b3) at (2,-3) {$b$};
						\node[draw,circle] (b4) at (2,-4.5) {$v_3$};
						\draw (a2) -- (b2);
						\draw (a2) -- (b3);
						\draw (a3) -- (b2);
						\draw (a3) -- (b3);
						\draw (a3) -- (b4);
						\draw (a4) -- (b3);
						\draw (a4) -- (b4);
				\end{tikzpicture}} &
				\parbox{.4\textwidth}{\centering
					\begin{tikzpicture}[scale=.75]
						\node[draw,circle] (a2) at (0,-1.5) {$v_1$};
						\node[draw,circle] (a3) at (0,-3) {$a$};
						\node[draw,circle] (a4) at (0,-4.5) {$v_2$};
						\node[draw,circle] (b2) at (2,-1.5) {$v_4$};
						\node[draw,circle] (b3) at (2,-3) {$b$};
						\node[draw,circle] (b4) at (2,-4.5) {$v_3$};
						\draw[color=red,line width=1.5pt] (a2) -- (b2);
						\draw (a3) -- (b2);
						\draw (a3) -- (b3);
						\draw (a4) -- (b3);
						\draw[color=red,line width=1.5pt] (a4) -- (b4);
						\draw (a2) -- (a3);
						\draw (a2) -- (b3);
						\draw (a3) -- (a4);
						\draw (a3) -- (b4);
						\draw (b2) -- (b3);
						\draw (b3) -- (b4);
						\draw[color=red,line width=1.5pt] (b2).. controls (3,-3).. (b4);
						\draw[color=red,line width=1.5pt] (a2).. controls (-1,-3).. (a4);
				\end{tikzpicture}}\\
				\vspace{.1in}
				\parbox{.4\textwidth}{\centering $G$} &  \vspace{.1in} \parbox{.4\textwidth}{\centering $G^2$}
			\end{tabular}
			\caption{A biconvex bipartite graph $G$ such that $G^2$ is not chordal}
			\label{fig:biconvex}
		\end{figure}
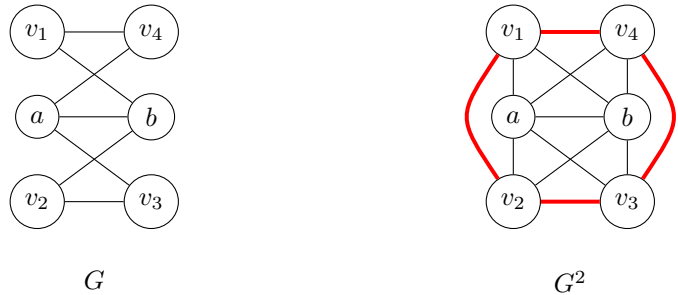
	We first note the following remark.
 \begin{remark}\textbf{\textit{Chordality is not a partite testable property for the squares of biconvex bipartite graphs:}}
    For a biconvex bipartite graph $G=(A,B,E)$, since both the partite sets have the consecutive ordering property, it is not difficult to infer from Lemma~\ref{obs:properinterval} that both the subgraphs $G^2[A]$ and $G^2[B]$  are proper interval graphs. But the squares of biconvex bipartite graphs need not be chordal.  See Figure~\ref{fig:biconvex} for an example of a biconvex bipartite graph $G$ for which $G^2$ is not chordal (the edges shown in red form an induced $C_4$). 
 \end{remark}
		However, we have the following corollary of Theorem~\ref{cycle:structure}, which implies that \emph{perfectness is a partite testable property for the squares of biconvex bipartite graphs}.
  		\begin{theorem}\label{thm:biconvex}
		The squares of biconvex bipartite graphs are $C_5$-free and hence perfect.
	\end{theorem}
		\begin{proof}
        Suppose not. Let $G=(A,B,E)$ be a biconvex bipartite graph. Then, the vertices in both the partite sets $A$ and $B$ admit consecutive orderings (see Definition~\ref{def:biconvex}). Suppose that $C$ is an induced cycle in $G^2$ of length five. Take $S=A$ and $\overline{S}=B$. Then, as the conditions in Corollary~\ref{corr:cyclepart} are satisfied, we get that $|V(C)\cap A|=2$. Now take $S=B$ and $\overline{S}=A$, then also the conditions in the same corollary are satisfied, and we have $|V(C)\cap B|=2$. Since $V(C)\subseteq A\cup B$, this implies that $|V(C)|\leq 4<5$, a contradiction. Therefore, we can conclude that the squares of biconvex bipartite graphs are $C_5$-free and hence perfect by Theorem~\ref{thm:c4free}\ref{item:c5free}.
       \end{proof}

\begin{remark}
   Let $\mathcal{C}$ denote the class of $C_5$-free squares of convex bipartite graphs, $\mathcal{C}_1$ denote the class of $C_4$-free squares of convex bipartite graphs, and $\mathcal{C}_2$ denote the class of squares of biconvex bipartite graphs. By Corollary~\ref{corr:cycleC_5} and Theorem~\ref{thm:biconvex}, respectively, we have $\mathcal{C}_1,\mathcal{C}_2\subseteq \mathcal{C}$. By a previous observation (recall Figure~\ref{fig:biconvex}), we have noted that the squares of biconvex bipartite graphs may contain $C_4$ as an induced subgraph. This implies that $\mathcal{C}_2\nsubseteq \mathcal{C}_1$. In fact, we can observe that $\mathcal{C}_1\nsubseteq \mathcal{C}_2$ as well. i.e. there exists convex bipartite graphs $G$ such that $G^2$ is $C_4$-free but $G$ is not a biconvex bipartite graph. See Figure~\ref{fig:convexC4free} for an example. It is not difficult to see that $G$ is a convex bipartite graph but not biconvex (since there are 3 vertices $b_i$, $i\in \{1,2,3\}$ in $B$ such that for each vertex $b_i$, the vertex $a_i\in A$ is an exclusive neighbor of $b_i$, whereas all the three vertices have the vertex $a$ as their common neighbor in $G$). On the other hand, $G^2$ is clearly $C_4$-free, implying that $\mathcal{C}_1\nsubseteq \mathcal{C}_2$. 
\end{remark}
\begin{figure}[h]
			\begin{tabular}{p{.4\textwidth}p{.4\textwidth}}
				\parbox{.4\textwidth}{\centering
					\begin{tikzpicture}[scale=.75]
						\node[draw,circle] (a2) at (0,-1.5) {$a$};
						\node[draw,circle] (a3) at (0,-3) {$a_1$};
						\node[draw,circle] (a4) at (0,-4.5) {$a_2$};
                        \node[draw,circle] (a5) at (0,-6) {$a_3$};
						\node[draw,circle] (b2) at (2,-1.5) {$b_1$};
						\node[draw,circle] (b3) at (2,-3) {$b_2$};
						\node[draw,circle] (b4) at (2,-4.5) {$b_3$};
						\draw (a2) -- (b2);
						\draw (a2) -- (b3);
						\draw (a2) -- (b4);
						\draw (a3) -- (b2);
						\draw (a4) -- (b3);
						\draw (a5) -- (b4);
						
				\end{tikzpicture}} &
				\parbox{.4\textwidth}{\centering
					\begin{tikzpicture}[scale=.75]
						\node[draw,circle] (a2) at (0,-1.5) {$a$};
						\node[draw,circle] (a3) at (0,-3) {$a_1$};
						\node[draw,circle] (a4) at (0,-4.5) {$a_2$};
                        \node[draw,circle] (a5) at (0,-6) {$a_3$};
						\node[draw,circle] (b2) at (2,-1.5) {$b_1$};
						\node[draw,circle] (b3) at (2,-3) {$b_2$};
						\node[draw,circle] (b4) at (2,-4.5) {$b_3$};
						\draw (a2) -- (b2);
						\draw (a2) -- (b3);
						\draw (a2) -- (b4);
						\draw (a3) -- (b2);
						\draw (a4) -- (b3);
						\draw (a5) -- (b4);
                        \draw (a2) -- (a3);
      
						\draw (b2) -- (b3);
						\draw (b3) -- (b4);
						\draw (b2).. controls (3,-3).. (b4);
						\draw (a2).. controls (-1,-3).. (a4);
                        \draw (a2).. controls (-1.5,-3).. (a5);
				\end{tikzpicture}}\\
				\vspace{.1in}
				\parbox{.4\textwidth}{\centering $G$} &  \vspace{.1in} \parbox{.4\textwidth}{\centering $G^2$}
			\end{tabular}
			\caption{A convex bipartite graph $G$ such that $G$ is not biconvex, but $G^2$ is $C_4$-free}
			\label{fig:convexC4free}
		\end{figure}
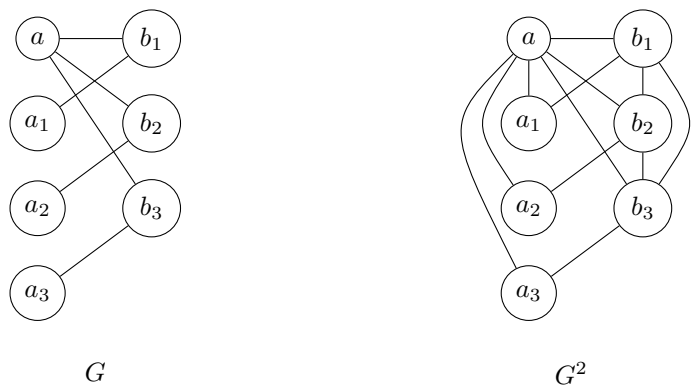
		\bibliographystyle{abbrv}
		\bibliography{references.bib}

\begin{thebibliography}{10}

\bibitem{agnarsson2003coloring}
G.~Agnarsson and M.~M. Halld{\'o}rsson.
\newblock Coloring powers of planar graphs.
\newblock {\em SIAM Journal on Discrete Mathematics}, 16(4):651--662, 2003.

\bibitem{Alon2002}
N.~Alon and B.~Mohar.
\newblock The chromatic number of graph powers.
\newblock {\em Combinatorics, Probability and Computing}, 11:1 -- 10, 2002.

\bibitem{valerio}
V.~Boncompagni, I.~Penev, and K.~V. skovi\'c.
\newblock Clique cutsets beyond chordal graphs.
\newblock {\em Electronic Notes in Discrete Mathematics}, 62:81--86, 2017.

\bibitem{booth1976}
K.~S. Booth and G.~S. Lueker.
\newblock Testing for the consecutive ones property, interval graphs, and graph
  planarity using $pq$-tree algorithms.
\newblock {\em Journal of Computer and System Sciences}, 13(3):335--379, 1976.

\bibitem{calamoneri2009}
T.~Calamoneri, S.~Caminiti, R.~Petreschi, and S.~Olariu.
\newblock On the $l(h, k)$-labeling of co-comparability graphs and circular-arc
  graphs.
\newblock {\em Networks}, 53(1):27--34, 2009.

\bibitem{choi2020}
I.~Choi, D.~W. Cranston, and T.~Pierron.
\newblock Degeneracy and colorings of squares of planar graphs without
  4-cycles.
\newblock {\em Combinatorica}, 40(5):625--653, 2020.

\bibitem{chudnovsky2006strong}
M.~Chudnovsky, N.~Robertson, P.~Seymour, and R.~Thomas.
\newblock The strong perfect graph theorem.
\newblock {\em Annals of mathematics}, pages 51--229, 2006.

\bibitem{cranston2022}
D.~W. Cranston.
\newblock Coloring, list coloring, and painting squares of graphs (and other
  related problems).
\newblock {\em arXiv preprint arXiv:2210.05915}, 2022.

\bibitem{dong2019}
W.~Dong and B.~Xu.
\newblock 2-distance coloring of planar graphs without 4-cycles and 5-cycles.
\newblock {\em SIAM Journal on Discrete Mathematics}, 33(3):1297--1312, 2019.

\bibitem{dvovrak2008}
Z.~Dvo{\v{r}}{\'a}k, P.~Nejedl{\`y}, and R.~{\v{S}}krekovski.
\newblock Coloring squares of planar graphs with girth six.
\newblock {\em European Journal of Combinatorics}, 29(4):838--849, 2008.

\bibitem{fraigniaud2020}
P.~Fraigniaud, M.~M. Halld\'{o}rsson, and A.~Nolin.
\newblock Distributed testing of distance-$k$ colorings.
\newblock In {\em Structural Information and Communication Complexity: 27th
  International Colloquium (SIROCCO)}, page 275–290, 2020.

\bibitem{golumbic}
M.~C. Golumbic.
\newblock {\em Algorithmic Graph Theory and Perfect Graphs (Annals of Discrete
  Mathematics, Vol 57)}.
\newblock North-Holland Publishing Co., 2004.

\bibitem{congestmodel}
M.~M. Halld\'{o}rsson, F.~Kuhn, and Y.~Maus.
\newblock Distance-2 coloring in the congest model.
\newblock In {\em Proceedings of the 39th ACM Symposium on Principles of
  Distributed Computing (PODC)}, page 233–242, 2020.

\bibitem{karapetian1980coloring}
I.~A. Karapetian.
\newblock On coloring of arc graphs.
\newblock {\em Dokladi of the Academy of Science of the Armenian SSR},
  70(5):306--311, 1980.

\bibitem{karthick2018coloring}
T.~Karthick and F.~Maffray.
\newblock Coloring (gem, co-gem)-free graphs.
\newblock {\em Journal of Graph Theory}, 89(3):288--303, 2018.

\bibitem{karthick2019}
T.~Karthick and F.~Maffray.
\newblock Square-free graphs with no six-vertex induced path.
\newblock {\em SIAM Journal on Discrete Mathematics}, 33(2):874--909, 2019.

\bibitem{kempe1879}
A.~B. Kempe.
\newblock On the geographical problem of the four colours.
\newblock {\em American journal of mathematics}, 2(3):193--200, 1879.

\bibitem{klemz2022}
B.~Klemz and G.~Rote.
\newblock Linear-time algorithms for maximum-weight induced matchings and
  minimum chain covers in convex bipartite graphs.
\newblock {\em Algorithmica}, 84(4):1064--1080, 2022.

\bibitem{kloks2009even}
T.~Kloks, H.~M{\"u}ller, and K.~Vu{\v{s}}kovi{\'c}.
\newblock Even-hole-free graphs that do not contain diamonds: a structure
  theorem and its consequences.
\newblock {\em Journal of Combinatorial Theory, Series B}, 99(5):733--800,
  2009.

\bibitem{KRAMER2008422}
F.~Kramer and H.~Kramer.
\newblock A survey on the distance-colouring of graphs.
\newblock {\em Discrete Mathematics}, 308(2):422--426, 2008.

\bibitem{Le2019hardness}
H.-O. Le and V.~B. Le.
\newblock Hardness and structural results for half-squares of restricted tree
  convex bipartite graphs.
\newblock {\em Algorithmica}, 81:4258--4274, 2019.

\bibitem{molloy2005bound}
M.~Molloy and M.~R. Salavatipour.
\newblock A bound on the chromatic number of the square of a planar graph.
\newblock {\em Journal of Combinatorial Theory, Series B}, 94(2):189--213,
  2005.

\bibitem{scott2020survey}
A.~Scott and P.~Seymour.
\newblock A survey of $\chi$-boundedness.
\newblock {\em Journal of Graph Theory}, 95(3):473--504, 2020.

\bibitem{thomassen2018}
C.~Thomassen.
\newblock The square of a planar cubic graph is 7-colorable.
\newblock {\em Journal of Combinatorial Theory, Series B}, 128:192--218, 2018.

\bibitem{tucker1975coloring}
A.~Tucker.
\newblock Coloring a family of circular arcs.
\newblock {\em SIAM Journal on Applied mathematics}, 29(3):493--502, 1975.

\bibitem{valencia2003revisiting}
M.~Valencia-Pabon.
\newblock Revisiting tucker's algorithm to color circular arc graphs.
\newblock {\em SIAM Journal on Computing}, 32(4):1067--1072, 2003.

\bibitem{van2003coloring}
J.~van~den Heuvel and S.~McGuinness.
\newblock Coloring the square of a planar graph.
\newblock {\em Journal of Graph Theory}, 42(2):110--124, 2003.

\bibitem{zhu2022}
J.~Zhu, Y.~Bu, and H.~Zhu.
\newblock Wegner's conjecture on 2-distance coloring for planar graphs.
\newblock {\em Theoretical Computer Science}, 926:71--74, 2022.

\end{thebibliography}
	\end{document}